\documentclass[runningheads,UKenglish,cleveref, autoref, thm-restate]{llncs}
\usepackage[T1]{fontenc}
\usepackage{graphicx}
\usepackage[T1]{fontenc}
\usepackage{graphicx}
\usepackage[dvipsnames]{xcolor}
\usepackage{amssymb}
\usepackage{amsmath}
\usepackage{mathtools}
\usepackage{todonotes}
\usepackage{diagbox}
\usepackage{tabularx,booktabs}
\usepackage[capitalise]{cleveref}
\newcolumntype{Y}{>{\centering\arraybackslash}X}

\newcommand{\N}{\mathbb{N}}

\def\ta{\mathtt{a}}
\def\tb{\mathtt{b}}
\def\tc{\mathtt{c}}
\def\td{\mathtt{d}}

\def\r{\operatorname{r}}

\def\m{\operatorname{m}}

\renewcommand{\epsilon}{\varepsilon}
\renewcommand{\phi}{\varphi}

\DeclareMathOperator{\arch}{ar}
\DeclareMathOperator{\ar}{ar}

\DeclareMathOperator{\alphabet}{alph}
\DeclareMathOperator{\al}{alph}
\DeclareMathOperator\ScatFact{SubSeq}

\DeclareMathOperator\SubWords{SubSeq}
\DeclareMathOperator{\letters}{alph}

\DeclareMathOperator{\PerfUniv}{PUniv}
\DeclareMathOperator{\MinPerfUniv}{MinPUniv}

\DeclareMathOperator{\Univ}{Univ}
\DeclareMathOperator{\inner}{in}

\DeclareMathOperator{\nextAlphPos}{next}
\DeclareMathOperator{\splitt}{split}

\DeclareMathOperator{\prevArray}{samePrev}
\DeclareMathOperator{\nextpos}{nextA}
\DeclareMathOperator{\jumpAhead}{{jumpFwd}}
\DeclareMathOperator{\jumpBack}{{jumpBack}}

\newcommand{\harpvecsign}{\scriptscriptstyle\leftharpoonup}
\newcommand{\harpoonvec}[2]{%
	\ifx\displaystyle#1\doalign{$\harpvecsign$}{#1#2}\fi
	\ifx\textstyle#1\doalign{$\harpvecsign$}{#1#2}\fi
	\ifx\scriptstyle#1\doalign{\scalebox{.6}[.9]{$\harpvecsign$}}{#1#2}\fi
	\ifx\scriptscriptstyle#1\doalign{\scalebox{.5}[.8]{$\harpvecsign$}}{#1#2}\fi
}
\newcommand{\doalign}[2]{%
	{\vbox{\offinterlineskip\ialign{\hfil##\hfil\cr#1\cr$#2$\cr}}}%
}

\def\nth#1{#1$^{\text{th}}$}

\newif\ifpaper
\paperfalse 

\usepackage{algorithm}
\usepackage{algpseudocode}
\newtheorem{observation}{Observation}

\begin{document}
	\title{Tight Bounds for the Number of Absent Subsequences}
	%
	%
	\author{ Duncan Adamson\inst{1} \and
		Pamela Fleischmann\inst{2} \and
		Annika Huch\inst{2} \and
		Florin Manea\inst{3} \and
		Paul Sarnighausen-Cahn\inst{3} \and
		Max Wiedenhöft\inst{2}
		}
	\authorrunning{D. Adamson et al.}
	%
	\institute{University of St Andrews, UK\\
		 \email{duncan.adamson@st-andrews.ac.uk}\and
		Department of Computer Science, Kiel University, Germany\\
		\email{\{fpa,ahu,maw\}@informatik.uni-kiel.de} \and
		Department of Computer Science, University of Göttingen, Germany\\
		\email{\{florin.manea,paul.cahn\}@cs.uni-goettingen.de}}
	\maketitle              

\begin{abstract}
A {\em subsequence} of a word $w$ is a word $u$ that can be obtained by deleting some letters from $w$ while maintaining the relative order of the remaining letters, e.g., $\mathtt{lala}$ is a subsequence of $\mathtt{alfalfa}$. A word, over some alphabet $\Sigma$, which has all possible words of length $\iota$ over $\Sigma$ as subsequences is called $\iota$-universal, and the largest $\iota$ for which this holds is called the universality index of $w$, and denoted $\iota(w)$. Moreover, words that are not subsequences of $w$ are called absent subsequences (AS) of $w$, and their investigation was started in (Kosche et al., 2022). In this paper, we present tight bounds 
on the number of AS of a given length $k$ among all words with the same universality index $\iota$. For both the lower and upper bound, we construct words that have, respectively, a minimal and maximal number of absent subsequences of the respective length $k$, and, in the case of the lower bound, we provide the exact number of missing subsequences as a closed form. Finally, we present efficient enumeration algorithms for the set of subsequences of given length of a word: we give a novel, optimal enumeration algorithm with output linear delay of this set of subsequences, with preprocessing time $O(|w|)$, which is further improved to an incremental enumeration algorithm with $O(1)$ delay of this set of subsequences, with preprocessing time $O(|w|)$.\looseness=-1
\end{abstract}


	%
	%
	%
	\section{Introduction}
	For a given word $w$, a \emph{subsequence} (\emph{scattered factor}) of $w$ is a 
word obtained by deleting letters from $w$, i.e.,  formally a word $u$ such that there exist  
indices $1 \leq i_1 < i_2 < \dots < i_{\vert u \vert} \leq \vert w \vert$ with $u = w[i_1] w[i_2] 
\cdots w[i_{\vert u \vert}]$ and $\vert u \vert$ (resp.~$\vert w \vert$) denotes the length of $u$ 
(resp.~$w$). For instance, $\mathtt{mega}, \mathtt{meat}$, and $\mathtt{gate}$ are subsequences 
of $w=\mathtt{pomegranate}$ while $\mathtt{goat}$ is not, because 
its letters do not occur in the correct order in $w$.
The relation between words and their subsequences has been studied in
logics \cite{DBLP:conf/lics/HalfonSZ17,10.1007/978-3-030-17127-8_20}, as well as language and automata theory \cite{DBLP:conf/isaac/AdamsonFHKMN23,DBLP:journals/ipl/KarandikarKS15,DBLP:journals/lmcs/KarandikarS19,simon1972hierarchies,DBLP:conf/automata/Simon75,zetzsche:LIPIcs.ICALP.2016.123}.
The set of subsequences of a word reveals information about existing and absent parts of the word, and thus are not only of theoretical but of practical interest, too. When examining discrete data, subsequences are often used to model corrupted data \cite{dress2005reconstructing,DBLP:journals/ijfcs/FleischmannLMNR21,DBLP:conf/dlt/Manuch99} related to the reconstruction problem. These problems appear in a wide variety of fields, such as 
formal software verification \cite{DBLP:conf/lics/HalfonSZ17,zetzsche:LIPIcs.ICALP.2016.123}. Subsequences have also been studied in the context of database theory \cite{artikis2017complex,Kleest-Meissner22,Kleest-Meissner23,SchmidSIGMOD} or with motivation coming from this field \cite{DayKMS22,Goettingen2023words,ManeaRS24}. 

In terms of theoretical work, the study of subsequences is strongly related to the 1972-work of Simon \cite{simon1972hierarchies}, where the famous congruence relation, now known as \emph{Simon's congruence}, regarding piecewise-testable languages is introduced. Two words $u$ and $v$ are called \emph{Simon $k$-congruent} w.r.t.~a natural number $k$ if they have the same set of subsequences of length at most $k$. For example, $\mathtt{aaba}$ and $\mathtt{abaa}$ are Simon $2$-congruent but not $3$-congruent.
For a more detailed overview of Simon's congruence, we refer to \cite[Chapter 6]{lothaire} by Sakarovich and Simon, as well as the surveys \cite{pin2019influence,Kosche2022SubsequenceSurvey}. Despite intense scrutiny~\cite{DBLP:conf/mfcs/FleischerK18,DBLP:conf/stacs/GawrychowskiKKM21} and several attempts toward giving the number of congruence classes,  this problem remains open and continues to attract interest \cite{DBLP:conf/fct/FleischmannHHN23,DBLP:conf/rp/FleischmannKKMNSW23,DBLP:journals/tcs/KimHKS23,DBLP:conf/dlt/KimHKS23,DBLP:journals/tcs/KimKH24}. Current developments inspect both the shortest absent subsequence of words \cite{DBLP:journals/fuin/KoscheKMS22} and relations to universality index on the other hand \cite{DBLP:conf/dlt/BarkerFHMN20,DBLP:conf/stacs/DayFKKMS21}. This notion of universality serves as a measure of the containment of all subsequences of up to a given length and originates from two sources: Karandikar et al. \cite{DBLP:journals/ipl/KarandikarKS15} introduced the notion of $k$-richness in the context of piecewise testable languages \cite{katai2012word,pach2018normal} 
which coincides with the notion of $k$-universality. The language of $k$-universal words is extensively studied and characterised in \cite{DBLP:conf/isaac/AdamsonFHKMN23,DBLP:conf/dlt/BarkerFHMN20,DBLP:conf/stacs/DayFKKMS21,DBLP:journals/tcs/Hebrard91}. 
One of the main insights concerning subsequence-universality is that the $k$-universality of words is tightly related to the {\em arch factorisation} by Hébrard   \cite{DBLP:journals/tcs/Hebrard91}.
One approach followed in \cite{DBLP:journals/tcs/FleischmannHHHMN23,DBLP:journals/fuin/KoscheKMS22} is to characterise absent subsequences of words in order to determine the index of Simon's congruence, i.e., parametrise the problem by the number of absent subsequences. 
Additionally, absent subsequences seem to naturally occur in rather practical scenarios such as in the context of reachability and avoidability problems~\cite{DBLP:journals/fuin/KoscheKMS22}.
In \cite{DBLP:journals/fuin/KoscheKMS22}, the authors gave an implicit characterisation and a tree-based representation of all shortest absent subsequences of a given word.\looseness=-1

\medskip

\noindent
\textbf{Own Contribution.}
We extend the ideas of \cite{DBLP:journals/tcs/FleischmannHHHMN23,DBLP:conf/fct/FleischmannHHN23,DBLP:journals/fuin/KoscheKMS22} to the analysis of absent subsequences, that are longer than the shortest ones. A first generalisation is to ask for the minimal and maximal number of absent subsequences of a word in relation to the number of arches $\iota$ and the size of the alphabet. Thus, we ask - given an alphabet  $\Sigma$, and positive integers $\iota$ and $k$ where $\iota < k$ - what is the structure of a word containing exactly $\iota$ arches with the greatest (resp. fewest) number of absent subsequences of length $k$? In this way, we build on the work of \cite{DBLP:conf/isaac/AdamsonFHKMN23} in looking at the properties of $k$-universality within languages. While \cite{DBLP:conf/isaac/AdamsonFHKMN23} considered the problem of deciding whether there exists a $k$-universal word in a given language, we look for the word that is closest (resp. furthest) to $k$-universality within the restricted language of $\iota$-universal words.
We give a tight characterisation for both. First, we describe the structure of words $w\in\Sigma^*$ that have the maximal number of existing subsequences, i.e., the minimal number of absent subsequences, of any given length $k\in\N$ among all words $w'\in\Sigma^*$ with the same number of arches (Section~\ref{sec:lowbound}). Additionally, we characterise the words fulfilling the above property and have the shortest possible length. After that, we exhibit a class of words $w_{\min}\in\Sigma^*$ with a minimal number of existing subsequences, i.e., a maximal number of absent subsequences among all words $w'\in\Sigma^*$ with the same number of arches (Section~\ref{sec:upbound}). Finally, we address the problem of enumerating the set of subsequences of given length of an input word (Section~\ref{sec:enum}). In particular, we give a novel, optimal enumeration algorithm with output linear delay of this set of subsequences with preprocessing time $O(|w|)$. This is then further improved to an incremental enumeration algorithm (where only the differences between two consecutive outputs are succinctly provided) with constant delay of this set of subsequences with preprocessing time $O(\vert w \vert)$.
%
%
%


	\section{Preliminaries}\label{sec:prelims}
	Let $\mathbb{N}$ be the set of all natural numbers, $\mathbb{N}_0 = \mathbb{N} \cup \{0\}$, 
$[n] = \{1,\ldots,n\}$, $[i,j]=[j]\backslash[i-1]$, and $[n]_0 := [n] \cup \{0\}$. An \emph{alphabet} 
$\Sigma$ is a non empty, finite set whose elements are called \emph{letters} or {\em symbols}.
Given $\sigma\in\N$, we fix  $\Sigma=\{\ta_1,\ldots,\ta_{\sigma}\}$.  A \emph{word} is a 
finite sequence of letters from $\Sigma$. Let $\Sigma^*$ be 
the set of finite words over $\Sigma$ and $\epsilon$ be the empty word. 
Set $\Sigma^+ := \Sigma^* \setminus \{\epsilon\}$. Let $w\in\Sigma^{\ast}$. 
For $n\in\N_0$, inductively define 
$w^0=\varepsilon$ and $w^n=ww^{n-1}$. 
The {\em length} of $w$ is the number of its letters; so $|\epsilon| = 0$. 
For all $k \in \mathbb{N}_0$ set $\Sigma^k := \{w \in \Sigma^* \mid |w| = k\}$ (define $\Sigma^{\leq k}$, $\Sigma^{\geq k}$ analogously). We denote $w$'s \nth{$i$} letter by $w[i]$ 
and $w[i:j]=w[i]\cdots w[j]$ if $i < j$,  and $\epsilon$ if $ i > j$ for all $i,j \in [\vert w \vert]$. 
Set $\letters(w) = \{\ta \in \Sigma \mid \exists i \in [|w|]: w[i] = \ta \}$ as $w$'s {\em alphabet} and 
for each $\ta \in \Sigma$ set $|w|_{\ta} = |\{i \in [|w|] \mid w[i]=\ta \}|$. The word $u \in 
\Sigma^*$ is a \emph{factor} (or infix) of $w$ if there exist $x,y \in \Sigma^*$ such that $w = xuy$.
We call $u$ a \emph{prefix} (\emph{suffix}) of $w$ if $x=\epsilon$ ($y = \epsilon$). 
We define the \emph{reverse} of $w$ by $w^R = w[|w|]\cdots w[1]$.
Let $<$ be a total order on $\Sigma$, and assume that $\ta_1<\ta_2<\ldots < \ta_\sigma$ are the letters of $\Sigma$; define $w_{\Sigma} \in \Sigma^{\sigma}$ as $w_{\Sigma}=\ta_1\cdots\ta_{\sigma}$.
For further definitions see \cite{lothaire}. 

\begin{definition}
	Let $w \in \Sigma^*$ and $n \in \mathbb{N}_0$. A word $u \in \Sigma^n$ is called a \emph{subsequence} of w - denoted as $u \in \ScatFact(w)$ - if there exist $v_1,\dots, v_{n+1} \in \Sigma^*$ such that $w = v_1 u[1] v_2 u[2] \cdots v_n u[n]$ $v_{n+1}$.
	Set $\ScatFact_k(w)=\{u\in\ScatFact(w)|\,|u|=k\}$.
\end{definition}

The words $\mathtt{cafe}$ and $\mathtt{ufo}$ are all subsequences of $\mathtt{cauliflower}$ but neither $\mathtt{flour}$ nor $\mathtt{row}$. Subsequences are tightly related to {\em Simon's congruence}. 

\begin{definition}
	Two words $w,v \in \Sigma^*$ are \emph{Simon congruent} w.r.t. $k\in\N_0$ (denoted $w \sim_k v$) if $\operatorname{\ScatFact}_{\ell}(w) = \operatorname{\ScatFact}_{\ell}(v)$ for all $\ell\leq k$. 
\end{definition}

Since $\ScatFact_k(w)\subseteq\Sigma^k$ holds for all $k\in\N_0$, determining the index of Simon's
congruence (i.e., the number of equivalence classes w.r.t.~$\sim_k$) can be split into the parametrised problem on determining how many subsequence sets 
- or equivalently how many different words - exist with $|\ScatFact_k(w)|=\sigma^k-m$ for all 
$m\in\N_0$ (cf. \cite{DBLP:conf/dlt/BarkerFHMN20,DBLP:conf/stacs/DayFKKMS21} for $m=0$).
We disregard all \(w \in \Sigma^{<k}\) because they form singleton classes. For all \(w\in\Sigma^{\geq k}\) it suffices to consider the set \(\ScatFact_k(w)\) instead of $\ScatFact_{\ell}(w)$ for all $\ell \leq k$ because all shorter subsequences are contained in a subsequence of length \(k\).

\begin{definition}
	For given $k,m\in\N_0$, a word $w \in \Sigma^*$ is called \emph{$m$-nearly $k$-universal} if $\vert \ScatFact_k(w) \vert $$= 
	\sigma^k-m$. 
	If $m=0$, the words are called $k$-universal and the set of all $k$-universal words over $\Sigma$ is $\Univ_{\Sigma,k}$. The largest $k\in\N_0$, such that a word is $k$-universal, is called the {\em universality index} $\iota(w)$.
\end{definition}


For an unary alphabet $\{a\}$, all words are of the form $\ta^{\ell}$ for $\ell\in\N_0$ and their subsequences can be determined easily.  Thus, we only consider here alphabets with two or more letters. Moreover, we assume $\Sigma=\letters(w)$ 
for a given $w$, if not stated otherwise.
One of the main tools used in the investigation of ($m$-nearly) $k$-universal words is the \emph{arch factorisation}, introduced by Hébrard \cite{DBLP:journals/tcs/Hebrard91}. 

\begin{definition}
	For $w \in \Sigma^*$ the \emph{arch factorisation} is $w = \ar_1(w) \cdots \ar_k(w) \r(w)$ for some $k \in \mathbb{N}_0$ with\\
	(a) $\iota(\ar_i(w))=1$ for all $i \in [k]$, \\
	(b) $\ar_i(w)[\vert \ar_i(w) \vert] \notin \letters(\ar_i(w)[1 : \vert \ar_i(w) \vert - 1 ])$ for all $i \in [k]$, and\\
	(c) $\letters(\r(w)) \subsetneq \Sigma$.\\
	The words $\ar_i(w)$ are called \emph{arches} and $\r(w)$ is the \emph{rest} of $w$.
	The {\em modus} $\m(w)$ is given by  $\ar_1(w)[\vert\ar_1(w)\vert] \cdots\ar_k(w)[\vert\ar_k(w)\vert]$, i.e., it is the word containing the unique last letters of each arch.
	The \emph{inner part of the $i^{th}$ arch of $w$}, denoted $\inner_i(w)$, is defined as $ \inner_i(w) = \ar_i(w)[1 : \vert \ar_i(w) \vert - 1 ]$.
\end{definition}

To visualise the arch factorisation in explicit examples we use parenthesis, e.g.,  $(\mathtt{aab})\cdot (\mathtt{bba}) \cdot \ta$ has two arches and the rest $\ta$. Here, the modus is $\tb\ta$ and the two inner parts are $\ta\ta$ and $\tb\tb$ resp.

\begin{definition}
	We call a word $w \in \Sigma^*$ \emph{perfect $k$-universal} if $\iota(w)=k$ and $\r(w) = \epsilon$. The set of all these words with $\letters(w)=\Sigma$ is denoted by $\PerfUniv_{\Sigma,k}$.
	Additionally, we call $w$ \emph{minimal perfect $k$-universal} if $w\in\PerfUniv_{\Sigma,k}$ and $w$ is of minimal length among all words in $\PerfUniv_{\Sigma,k}$.
	The set of those words over $\Sigma$ is denoted by $\MinPerfUniv_{\Sigma,k}$.
\end{definition}

For the algorithmic results we use the standard 
computational model RAM with logarithmic word-size (see, e.g., \cite{KarkkainenSB06}). We also 
follow a standard assumption from stringology, that if $w$ is the 
input word for our algorithms, then we assume that $w$ is over an alphabet $\Sigma=\{ 1,2, \ldots, \sigma\}$ for some $\sigma\leq |w|$.


	
	
	\section{Lower Bound for the Number of Absent Subsequences}\label{sec:lowbound}
	In this section, we consider the following problem: for a specific alphabet $\Sigma$, universality index $\iota\in\N$, and number $k$, what is the minimum number of absent subsequences of length $k$ for any word $w$ with $\iota(w) = \iota$? Our approach is to construct a word that contains as many subsequences as possible (and thus as few absent ones as possible) and compute its number of absent subsequences. Additionally, we are particularly interested in finding the shortest words fulfilling this property. We start by understanding which words are always absent subsequences from $\ScatFact_k(w)$ (for $w$ with the above properties) and show that there actually exist words in which only these words are absent, providing a tight lower bound on the number of absent subsequences.

\ifpaper
\else 
\begin{lemma}\label{lemma:atleast-absent-sf}
    Let $ w \in \Sigma^*$ and let $k \in \N$ with $k > \iota(w)$. Then, there exists $\ta \in \Sigma\setminus \al(\r(w))$ such that, for every $v \in \Sigma^k$ with $\m(w)\ta \in \ScatFact_{\iota(w)+1}(v)$, we have $v\notin\ScatFact_k(w)$.
\end{lemma}

\begin{proof}
    Note first that $\alphabet(\r(w)) \neq \Sigma$. Therefore, there exists some $\ta\in\Sigma$ with $\ta\notin\alphabet(\r(w))$. As every letter of $\m(w)$ is unique within its arch, we know that any $v\in\ScatFact_k(w)$
    with $\m(w) \ta \in \ScatFact_{\iota(w)+1}(v)$ implies $\ta\in\r(w)$ - a contradiction.\qed
\end{proof}

\fi



\ifpaper
\else
Lemma \ref{lemma:atleast-absent-sf} gives us the number of subsequences of length $k$ which are missing in any word, i.e., a lower bound on $m$. We now show that this bound is tight.
\fi

\ifpaper
\else
\begin{lemma}\label{lemma:min-absent-sf-strucure}
    Given $\iota,k\in\N$ such that $k > \iota$, and $\ta\in\Sigma$, there exists $w\in\Sigma^*$ with $\iota(w) = \iota$ such that we have $\ScatFact_k(w)=\{v\in \Sigma^k\mid \m(w)\ta$ is not a subsequence of $v\}$
\end{lemma}

\begin{proof}
Let $k,\iota\in\N$ and set $\Sigma^-=\{\ta_1,\ldots,\ta_{\sigma-1}\}$. W.l.o.g. we assume $\ta=\ta_{\sigma}$ (otherwise we could permute $\Sigma$).
    Set $w = (w_{\Sigma^-}^{k-\iota}\ta_\sigma)^{\iota}w_{\Sigma^-}^{k-\iota}$ with $\iota(w)=\iota$, $\m(w)=\ta_{\sigma}^{\iota}$. 
    Let $v\in\Sigma^k\backslash\ScatFact_k(w)$. By the construction of $w$, each word $u\in\Sigma^k$ with $|u|_{\ta_{\sigma}}\leq\iota$ is a subsequence of $w$. Thus, we have $|v|_{\ta_{\sigma}}>\iota$ and we have $\m(w)\ta\in\ScatFact_k(v)$.\qed
\end{proof}

\fi

\ifpaper
\else
Together, these lemmata yield the following proposition.
\fi

\begin{proposition}\label{corollary:min-absent-sf-main-property}
    For $k\in\N$, $w \in \Sigma^*$ with $k > \iota(w)$, the next are equivalent:\\
    1) For all  $w' \in \Sigma^*$  with $\iota(w')=\iota(w)$, we have $\vert \ScatFact_k(w') \vert \leq \vert \ScatFact_k(w) \vert$ (i.e., $w$ is a word with a maximal number of subsequences of length $k$, over all words with the same universality index).\\
    2) There exists $\ta\in \Sigma$ such that $\Sigma \setminus \al(\r(w))=\{\ta\}$  and $\ScatFact_k(w)=\{v\in \Sigma^k\mid \m(w)\ta$ is not a subsequence of $v\}$.
    
    
%
%
%
%

\end{proposition}

\ifpaper
\else
\begin{proof}
Let us first prove that $1\Rightarrow 2$. 

Let $\iota=\iota(w)$. Clearly, there exists $\ta\in \Sigma\setminus \al(\r(w))$ and assume $\Sigma=\{\ta_1,\ldots,\ta_\sigma\}$ with $\ta_\sigma=\ta$. 

We define $w' = (\inner_1(w))^k\ar_1(w) \cdots (\inner_\iota(w))^k\ar_\iota(w) (\ta_1\cdots \ta_{\sigma-1})^k\r(w)$. It is immediate to note that $\iota(w')=\iota(w)=\iota$, $m(w')=m(w)$, and $\al(\r(w))=\Sigma\setminus \{\ta\}$. Clearly, as $w$ is a subsequence of $w'$, we have that $\ScatFact_k(w)\subseteq \ScatFact_k(w')$; but $w$ is a word with a maximal number of subsequences of length $k$, over all words with the same universality $\iota$, so $\ScatFact_k(w)= \ScatFact_k(w')$. 

Let us show that $\ScatFact_k(w')=\{v\in \Sigma^k\mid \m(w)\ta$ is not a subsequence of $v\}$. As $\m(w)\ta$ is not a subsequence of $w'$, we have that $\ScatFact_k(w')\subseteq \{v\in \Sigma^k\mid \m(w)\ta$ is not a subsequence of $v\}$. Now, consider $v$ to be a word of length $k$ that does not have $y=\m(w)\ta$ as a subsequence. We can write $v= v_1 \m(w)[1] \cdots v_\ell \m(w)[\ell] v_{\ell+1}$, where $m(w)[i] \notin \al(v_i)$, for $i\in [\ell]$, and, if $\ell=\iota$, then $\ta \notin \al(v_{\ell+1})$, or, if $\ell < \iota$, then $\m(w)[\ell+1] \notin \al(v_{\ell+1})$. In both cases, as $|v_i|\leq k$ for all $i\in [\ell+1]$, we have that $v_\ell \m(w)[\ell]$ is a subsequence of $(\inner_\ell(w))^k\ar_\ell(w)$. If $\ell=\iota$, then $v_{\ell+1}$ is a subsequence of $(\ta_1\cdots \ta_{\sigma-1})^k\r(w)$, and, if $\ell < \iota$, then $v_{\ell+1}$ is a subsequence of $(\inner_{\ell+1}(w))^k\ar_{\ell+1}(w)$. So, $v$ is a subsequence of $w'$. 

Now, assume that there exists $\tb\in \Sigma\setminus \al(\r(w))$ with $\tb\neq \ta$. Then, $\m(w)\tb$ is a subsequence of $(\inner_1(w))^k\ar_1(w) \cdots (\inner_\iota(w))^k\ar_\iota(w) (\ta_1\cdots \ta_{\sigma-1})$, so there is a word $y$ of length $k$, that has $\m(w)\tb$ as prefix which is a subsequence of $w' = (\inner_1(w))^k\ar_1(w) \cdots (\inner_\iota(w))^k\ar_\iota(w) (\ta_1\cdots \ta_{\sigma-1})^k\r(w)$. But $y$ cannot be a subsequence of $w$, which is a contradiction with the fact that $\ScatFact_k(w)= \ScatFact_k(w')$. So the statement holds. 

We can now prove that $2\Rightarrow 1$. 

Note that for some $n\geq k$, the number of words of length $n$, over $\Sigma$, which have a word $v$, of length $k$, as a subsequence equals the number of words of length $n$, over $\Sigma$, which have a word $v'\neq v$, of length $k$, as a subsequence. 

Let us assume, for the sake of a contradiction, that there is a word $w'$, with $\iota(w')=\iota(w)$, which has strictly more subsequences of length $k$ than $w$. After a potential renaming of the letters of $w'$, we can assume that $\ta\notin \al(\r(w'))$. Therefore, $w'$ has no subsequence of length $k$ which has $\m(w')\ta$ as a subsequence. So, $\ScatFact_k(w')\subseteq \Sigma^k \setminus \{v\in \Sigma^k\mid \m(w')\ta$ is not a subsequence of $v\}$. Thus, $|\ScatFact_k(w')|\leq |\Sigma^k| - |\{v\in \Sigma^k\mid \m(w')\ta$ is not a subsequence of $v\}|=  |\Sigma^k| - |\{v\in \Sigma^k\mid \m(w)\ta$ is not a subsequence of $v\}|=|\ScatFact_k(w)|$, a contradiction to our assumtion that $w'$ has strictly more subsequences of length $k$ than $w$. Our statement follows. 

Based on Lemma \ref{lemma:min-absent-sf-strucure}, one can note that a word $w$, with $\iota(w)=\iota$, with a maximal number of subsequences of length $k$, over all words with the same universality, can be defined as $w=((\ta_1 \cdots \ta_{\sigma-1})^{k-\iota(w)} \ta_\sigma)^\iota \cdot$  $(\ta_1 \cdots \ta_{\sigma-1})^{k-\iota(w)}$. \qed
 \end{proof}

\fi

Intuitively, for all words $w$ that are $n$-universal for some $n\in\N$, we know that their rest $\r(w)$ always misses at least one letter $\ta$. Additionally, we know that its modus $\m(w)$ occurs only once in $w$ and ends right before $\r(w)$. Hence, all words $v$ of length $k > |n| = |\m(w)|$ that have $\m(w)\ta$ as a subsequence, can in fact not be a subsequence of $w$. In fact, in the framework of Proposition \ref{corollary:min-absent-sf-main-property}, we can even construct a word $w_{\max}$ in which only subsequences $v\in\Sigma^k$ with $\m(w_{\max})\ta\in\ScatFact(v)$ are absent subsequences of length $k$, i.e., $v\notin\ScatFact_k(w_{\max})$. Assume $\Sigma=\{\ta_1,\ldots,\ta_\sigma\}$ with $\ta_\sigma = \ta$. Consider 
$$ w_{\max}=((\ta_1 \cdots \ta_{\sigma-1})^{k-\iota(w)} \ta)^\iota \cdot (\ta_1 \cdots \ta_{\sigma-1})^{k-\iota(w)}.$$ 
By Proposition \ref{corollary:min-absent-sf-main-property}, $w_{\max}$ must have the maximal number of subsequences of length $k$, i.e., the minimal number of absent subsequences of length $k$, of all words $w\in\Sigma^*$ with $\iota(w) = \iota(w_{\max}) = n$.
In general, assume $w$ to be a word with a maximal number of subsequences of length $k$ over all words with the same universality index $\iota(w)$. By Proposition \ref{corollary:min-absent-sf-main-property} and the existence of $w_{\max}$, we know that all absent subsequences of $w$ contain $\m(w)\ta$ as a subsequence themselves. Hence, we can derive a tight bound on the minimum number of absent subsequences of $w$. 

\begin{proposition}
    \label{proposition:min-absent-sf-count-max-factors}
    Let $k\in\N_0$ and $w \in \Sigma^*$ with $k > \iota(w)$ such that, for all $w' \in \Sigma^*$ with $\iota(w) = \iota(w')$, we have $\vert \ScatFact_k(w) \vert \geq \vert \ScatFact_k(w') \vert$. Then,
    $\vert \ScatFact_k(w) \vert = \sum_{j \in [0, \iota(w)]} \binom{k}{j} (\sigma-1)^{k - j}.$
\end{proposition}
\ifpaper
\else
\begin{proof}
	By Proposition~\ref{corollary:min-absent-sf-main-property}, we must determine the cardinality of 
	$\ScatFact_k(w)=\{v\in\Sigma^k|\,\m(w)\ta\not\in\ScatFact_{\iota(w)+1}(v)\}$ for some $\ta \in \Sigma$. Thus, we have to count the number of words of $\Sigma^k$ which do not contain a given $v'\in\Sigma^{\iota(w)+1}$ as a subsequence. We do so by considering words $v$ in $\Sigma^k$ such that some proper prefix of $v'$ is not a subsequence of $v$.
	First, the number of words that contain no prefix of $v'$ (equivalently, only the $0$ length prefix corresponding to the empty word) is $(\sigma - 1)^k$, corresponds to the set of words that do not contain the symbol $v'[1]$. 
	To count the number of words that do not contain the prefix $v'[1: 2]$, but do contain the prefix $v'[1]$, observe that any such word $v$ must contain some prefix $v[1:\ell]$ that does not contain the symbol $v'[1]$, followed the suffix $v'[1] v[\ell + 2:k]$ where $v[\ell + 2: k]$ does not contain the symbol $v'[2]$.  As there are $k$ possible positions at which $v'[1]$ could occur, and $(\sigma - 1)^{k - j}$ possible values of the remaining symbols for each such position $\ell \in [k]$, we have $(\sigma - 1)^{k - j} k$ possible words containing $v'[1]$ as a subsequence, and not $v'[1: 2]$.\looseness=-1
	
	In general, consider the number of words in $\Sigma^k$ that contain $v'[1: j]$ as a subsequence, but not $v'[1: j + 1]$. Observe that any such word $v$ can be expressed as $v[1: \ell_1 - 1] v'[1] v[\ell_1 + 1: \ell_2 - 1] v'[2] \cdots v'[j - 1] v[\ell_{j - 1} + 1: \ell_j - 1] v'[j] v[\ell_{j} + 1: k]$, with $v[\ell_{x}] = v'[x]$, $\forall x \in [j], v'[j+1]\not\in\letters(v[\ell_j+1:k])$.\looseness=-1
	
	For any given set of indices $\ell_1, \ell_2, \dots, \ell_j$ we have for the number of different $v\in\Sigma^k$ 
	\[
	\prod_{x \in [j + 1]}(\sigma - 1)^{\ell_{x} - 1 - \ell_{x - 1}} = (\sigma - 1)^{k+1 - 0} \prod_{x \in [j + 1]} (\sigma-1)^{-1} = (\sigma - 1)^{k-j}, 
	\]
	with $\ell_0 = 0$, $\ell_{j + 1} = k+1$. 
	There are $\binom{k}{j}$ possible indices $ 1 \leq \ell_1 < \dots < \ell_{j} \leq k$, i.e., $\binom{k}{j} \sigma^{k - j}$ such words. By extension, $\sum_{j \in [0, 
		\iota(w)]} \binom{k}{j} \sigma^{k - j}$ is the number of subsequences of $w$.\qed
\end{proof}
\fi

\begin{corollary}
    \label{corollary:min-absent-sf-count}
    Let $k\in\N_0$, $w \in \Sigma^*$ with $k > \iota(w)$ such that, for all $w' \in \Sigma^*$ with $\iota(w) = \iota(w')$, $\vert \ScatFact_k(w) \vert \geq \vert \ScatFact_k(w') \vert$ holds. The number of absent subsequences of length $k$ of $w$ is $\sigma^k - \sum_{j \in [0, \iota(w)]} \binom{k}{j} (\sigma-1)^{k - j}.$
\end{corollary}


\begin{example}\label{example:max-SF-word}
    Let $\Sigma := \{\ta,\tb,\tc\}$, $k = 4$, and $\iota = 2$. Then $w = (\ta\tb\ta\tb\tc)\cdot(\ta\tb\ta\tb\tc)\cdot\ta\tb\ta\tb$ is a word with a minimal number of absent subsequences of length $k$ among all words with universality index $\iota$.
    Note that exactly those words of length $4$ that are absent subsequences of $w$ have $\tc\tc\tc$ as a subsequence. 
\end{example}

Let $\iota, k\in\N$ with $k > \iota$, 
and consider $w\in\Sigma^*$, with $\iota(w) = \iota$, such that for all $w'\in\Sigma^{\geq |w|}$ with $\iota(w) = \iota(w')$ we have 
$|\ScatFact_k(w)| \geq |\ScatFact_k(w')|$ 
and, for all $w'\in\Sigma^{< |w|}$ with $\iota(w) = \iota(w')$ we have 
$|\ScatFact_k(w)| > |\ScatFact_k(w')|$. Then $w$ is the shortest word over $\Sigma$ which has $ \sum_{j \in [0, \iota(w)]} \binom{k}{j} (\sigma-1)^{k - j}$ (i.e., the maximum number of) subsequences of length $k$, among all words with universality index~$\iota$.\looseness=-1

We can now determine the structure of these shortest words.
\ifpaper
We first show that in such a word $w$, the $i^{th}$ arch fulfils $\inner_i(w) \in \Univ_{(\Sigma \setminus \{ \m(w)[i] \}), k - \iota}$, for all $i\in [\iota]$, and also $\r(w) \in \Univ_{(\Sigma \setminus \{ \ta \}), k - \iota}$, for some $\ta\in \Sigma$. From this, and the example-word $w_{\max}$ above, we can first derive their length.
\else 
We begin with two structural properties. 
\fi

\ifpaper
\else
\begin{lemma}\label{lemma:min-abssent-arches-k-iota-universal}
Let $k\in\N$ and $w\in\Sigma^{\ast}$ with $k>\iota(w)$ such that $\vert \ScatFact_k(w) \vert \geq \vert \ScatFact_k(w') \vert$ for every $w' \in \Sigma^*$ with $\iota(w') = \iota(w)$. 
Then,  for all $i \in [\iota(w)]$
\[
\inner_i(w) \in \Univ_{(\Sigma \setminus \{ \m(w)[i] \}), k - \iota(w)}.
\]
\end{lemma}

\begin{proof}
	Suppose that there exists some $i \in [\iota(w)]$ such that 
	\[
	\inner_i(w) \notin \Univ_{(\Sigma\setminus\{\m(w)[i]\}),k-\iota(w)}.
	\]
	Then there exists some $u\in(\Sigma\setminus\{\m(w)[i]\})^{k-\iota(w)}$
	such that 
	\[
	u\notin\ScatFact_{k-\iota(w)}(\inner_i(w)).
	\]
	Now, observe that there exists some $\tb\in\Sigma$ such that $\tb\notin\alphabet(\r(w))$.
	Set $v = \m(w)[1:i-1]u \m(w)[i+1:\iota(w)]\tb$. 
	Then, $v\notin\ScatFact_k(w)$ since $u$ is not contained in the \nth{$i$} arch. Hence we have a contradiction to \Cref{corollary:min-absent-sf-main-property} as $\m(w)\tb\notin\ScatFact_{\iota(w)+1}(v)$. \qed
\end{proof}

\fi




\ifpaper
\else
\begin{lemma}\label{remark:min-absent-rest-k-iota-universal}
    Let $k\in\N$, $w\in\Sigma^{\ast}$ with $k>\iota(w)$ such that $\vert \ScatFact_k(w) \vert \geq \vert \ScatFact_k(w') \vert$ for every $w' \in \Sigma^*$ with $\iota(w') = \iota(w)$. 
    Then, we have for some $\ta\in\Sigma$ 
    \[
    \r(w) \in \Univ_{(\Sigma \setminus \{ \ta \}), k - \iota(w)}.
    \]
\end{lemma}

\begin{proof}
	By Proposition \ref{corollary:min-absent-sf-main-property} we know that there exists some $\ta\in\Sigma$ such that
	for all $v\in\Sigma^k$ with $v\notin\ScatFact_k(w)$ we have $\m(w)\ta\in\ScatFact_{\iota(w)+1}(v)$.
	Suppose $\r(w) \notin \Univ_{(\Sigma \setminus \{ \ta \}), k - \iota(w)}$.
	Then there exists some $u\in\Sigma^{k-\iota(w)}$, $|u|_\ta = 0$, and $u\notin\ScatFact_{k-\iota(w)}(\r(w))$.
	By that, we know that $\m(w)u\notin\ScatFact_k(w)$. But this is a contradiction to \Cref{corollary:min-absent-sf-main-property}, as $\m(w)\ta\notin\ScatFact_{\iota(w)+1}(v)$.\qed
\end{proof}
\fi

\ifpaper
\else
This gives us a lower bound on the size of each arch $\arch_i(w)$ and the rest $\r(w)$ of a word $w\in\Sigma^*$. We now show that this bound is tight.
\fi

\ifpaper
\else
\begin{lemma}\label{lemma:min-absent-shortest-upper-bound}
    Let $k, \iota\in \N$ with $k > \iota$. There exists $w\in\Sigma^*$ with $\iota(w)=\iota$ and $\vert w \vert = (\sigma-1)\cdot(k-\iota)\cdot(\iota+1)+\iota$ such that, for all 
    $w'\in\Sigma^*$ with $\iota(w)  = \iota(w')$, $\vert\ScatFact_k(w')\vert \leq \vert\ScatFact_k(w)\vert$.
\end{lemma}

\begin{proof}
	The proof is similar to that of \Cref{lemma:min-absent-sf-strucure}.
	Let $\Sigma^-=\Sigma\backslash\{\ta_{\sigma}\}$ and $w = (w_{\Sigma^-}^{k-\iota} \ta_\sigma)^{\iota}w_{\Sigma^-}^{k-\iota}$.
	Note that for all $v\in\Sigma^k$ with $\m(w)\ta_\sigma = \ta_{\sigma}^{\iota+1} \in \ScatFact_{\iota+1}(v)$ we have $v \notin \ScatFact_k(w)$. Let $v'\in\Sigma^k$ such that  $\m(w)\ta_\sigma \notin \ScatFact_{\iota+1}(v')$. Then $\vert v'\vert_{\ta_\sigma} < \iota + 1$ as $\m(w) = \ta_{\sigma}^{\iota}$.
	As $\vert v' \vert = k$, $v'\in\ScatFact_k(w)$ and hence the claim follows from \Cref{corollary:min-absent-sf-main-property}.\qed
\end{proof}
\fi

\ifpaper
\else 
We can conclude the previous results in the following proposition.
\fi

\begin{proposition}\label{corollary:min-absent-shortest-word-length-tight}
    For $\iota, k\in\N$ with $k > \iota$, let $w$ be the shortest word over $\Sigma$ which has a maximum number of subsequences of length $k$, among all words with universality index $\iota$.
    Then, we have $ |w| = (\iota+1)\cdot(\sigma -1)\cdot(k-\iota) + \iota.$
\end{proposition}
\ifpaper
\else
\begin{proof}
	Let $\Sigma^-=\Sigma\backslash\{\ta_{\sigma}\}$ and $w = (w_{\Sigma^-}^{k-\iota} \ta_\sigma)^{\iota}w_{\Sigma^-}^{k-\iota}$ as explained before.
	$|w| = (\iota+1)\cdot(\sigma -1)\cdot(k-\iota) + \iota$. 
	Suppose there exists some $w'\in\Sigma^*$ with $\iota(w) = \iota(w')$
	and $|w'| < |w|$ such that $|\ScatFact_k(w')| \geq |\ScatFact_k(w)|$.
	Then either $|\r(w')| < |\r(w)|$ or $|\inner_i(w')| < |\inner_i(w)|$ for some $i\in[\iota]$.
	By construction outlined before 
	, $\r(w),\inner(w)\in\MinPerfUniv_{\Sigma^-,k-\iota}$. So, either $\r(w')\notin\Univ_{\Sigma^-,k-\iota}$ or
	$\inner_i(w')\notin\Univ_{\Sigma^-,k-\iota}$. Hence, by the above discussion on the form of the structure of every arch and the rest of words with a maximal amount of subsequences 
	we get that $|\ScatFact_k(w')| < |\ScatFact_k(w)|$ holds, which is a contradiction.\qed
\end{proof}
\fi

\Cref{corollary:min-absent-shortest-word-length-tight} gives a tight bound on the length of the shortest words with a minimal number of absent subsequences, of fixed universality,
\ifpaper
allowing us to also identify the precise structure of such words, in \cref{corollary:min-absent-arch-rest-puniv-iota}. Ultimately, this leads to the main result of this section, namely \cref{prop:equiv}. 
\else
We can also strengthen the result from \Cref{lemma:min-abssent-arches-k-iota-universal} by the following statement that follows from \Cref{corollary:min-absent-shortest-word-length-tight}
and \Cref{lemma:min-abssent-arches-k-iota-universal,remark:min-absent-rest-k-iota-universal}.
\fi

\begin{proposition}\label{corollary:min-absent-arch-rest-puniv-iota}
    For $\iota, k\in\N$ with $k > \iota$, let $w$ be a shortest word over $\Sigma$ with a maximum number of subsequences of length $k$, among all words with universality index $\iota$. Assume $\ta\notin\al(r(w))$ is the unique missing letter in the rest.
    Then, for all $i\in[\iota]$ and the corresponding arch $\ar_i(w)$, we have that its inner part $\inner_i(w) \in \MinPerfUniv_{(\Sigma\setminus\{\m(w)[i]\}),k-\iota}$,
    and $\r(w) \in \MinPerfUniv_{(\Sigma\setminus\{\ta\}),k-\iota}$.
\end{proposition}

\ifpaper
\else
We start by showing that the word of shortest length with a minimal number of absent subsequences has a rest in which only the last letter of the modus is missing.
\fi 

\ifpaper
\else
\begin{lemma}\label{lemma:min-absent-rest-no-last-modus}
    For $\iota, k\in\N$ with $k > \iota$, let $w$ be the shortest word over $\Sigma$ which has a maximum number of subsequences of length $k$, among all words with universality index $\iota$.
    Then $\alphabet(\r(w)) = \Sigma \setminus \{\m(w)[\iota]\}$.
\end{lemma}

\begin{proof}
	By Corollary \ref{corollary:min-absent-arch-rest-puniv-iota}, $\r(w)\in\MinPerfUniv_{(\Sigma\setminus\{\ta\}),k-\iota}$
	for some $\ta\in\Sigma$, and $\inner_{\iota}(w)\in\MinPerfUniv_{\Sigma\setminus\{\m(w)[\iota]\},k-\iota}$.
	Suppose, for the sake of a contradiction, that $\ta\neq\m(w)[\iota]$.
	Consider $v = \m(w)[1:\iota-1]\ta^{k-\iota+1}$.
	Then $\vert v\vert  = k$, but $v\notin\ScatFact_k(w)$ as $\vert \arch_{\iota}(w)\r(w)\vert _{\ta} = k-\iota$, because
	$\vert \r(w)\vert _\ta = 0$ and $\m(w)[\iota] \neq \ta$. By  \Cref{corollary:min-absent-sf-main-property}, however, we have that $\m(w)\ta\in\ScatFact_k(v)$ must hold (as $w$ has a maximal number of subsequences among all words with the same universality), which is a contradiction to the choice of $v$.
	Hence, $\alphabet(\r(w)) = \Sigma\setminus\{\m(w)[\iota]\}$.\qed
\end{proof}
\fi

\ifpaper
\else
Next, we get that the modus of words of the shortest length with a minimal number of absent subsequences has to be unary.
\fi

\ifpaper
\else
\begin{lemma}\label{lemma:min-abs-shortest-word-unary-modus}
    For $\iota, k\in\N$ with $k > \iota$, let $w$ be the shortest word over $\Sigma$ which has a maximum number of subsequences of length $k$, among all words with universality index $\iota$.
    Then $\vert \alphabet(\m(w))\vert  = 1$.
\end{lemma}

\begin{proof}
	Suppose $\vert \alphabet(\m(w))\vert  > 1$.
	Let $i\in[\iota]$ be the largest number such that $\m(w)[i] \neq \m(w)[\iota]$, i.e.
	for all $j\in[i+1:\iota]$ we have $\m(w)[j] = \m(w)[\iota]$.
	Let $a=\m(w)[\iota]$. Then, $\m(w)[i] \neq \ta$.
	By \Cref{lemma:min-absent-rest-no-last-modus}, $\alphabet(\r(w)) = \Sigma\setminus\{\ta\}$.
	Consider $v\in\Sigma^*$ such that 
	\[
	v = \m(w)[1:i-1]\ta^{k-\iota+1}\m(w)[i+2:\iota]\ta.
	\]
	Thus, we get $\vert v\vert  = k$ and we have 
	$v = m(w)[1:i-1]\ta^{k-i+1}$.
	As $\inner_i(w) \in \MinPerfUniv_{(\Sigma\setminus\{\m(w)[i]\}),k-\iota}$ and $\m(w)[i]\neq \ta$, $\inner_i(w)$ contains exactly 
	$k-\iota$ occurrences of $\ta$. 
	Since $\ta\notin\alphabet(\r(w))$ we have $|\ar_{i+1}(w)\cdots\ar_{\iota}(w)\r(w)|_{\ta}=
	\iota-(i+1)+1=\iota-i$ and, therefore, there are $k-i$ occurrences of $\ta$ in the suffix of $w$ starting with the \nth{$i$} arch. This leads to $v\not\in\ScatFact_k(w)$ since $k-i<k-i+1$.
	But, $\m(w)\ta \notin \ScatFact_{\iota+1}(v)$ as $\m(w)[i] \notin \alphabet(v[i:\vert v\vert ])$, contradicting \Cref{corollary:min-absent-sf-main-property}.
	Hence, the modus of $w$ is unary.\qed
\end{proof}

\fi


\begin{theorem}\label{prop:equiv}
    For $\iota, k\in \N$, $w\in\Sigma^{\ast}$ with $k > \iota=\iota(w)$, the next are equivalent:\\
    1) $w$ is the shortest word over $\Sigma$ with a maximum number of subsequences of length $k$, among all words with universality index $\iota$.\\
    2) there exists some $\ta\in\Sigma$ such that $\m(w) = \ta^\iota$ and,
    for all $i\in[\iota]$, we have 
    $\inner_i(w),\r(w) \in \MinPerfUniv_{\Sigma',k-\iota}$ for $\Sigma' = \Sigma\setminus\{\ta\}$.
\end{theorem}
\ifpaper
\else
\begin{proof}
	This first direction follows according to  \Cref{corollary:min-absent-arch-rest-puniv-iota}. 
	The other direction follows simply by constructing such a word and seeing that we obtain the shortest possible
	length of words that fulfil all properties. 
	\qed
\end{proof}
\fi

Further, we are intested in the total number of shortest words $w$ with a minimal number of absent subsequences of length $k$ over all words that are $\iota(w)$-universal. First, consider the binary case $\Sigma=\{\ta,\tb\}$. Choose w.l.o.g. $\m(w)=\ta^{\iota}$. Thus, the inner part of each arch and the rest are unary words, i.e., $w=(\tb^{k-\iota}\ta)^{\iota}\tb^{k-\iota}$. Since the choice of $\ta$ fixes the remaining letters, for any $(k,\iota)$-combination, we have exactly two shortest words with a minimal number of absent subsequences. For alphabets with $\sigma$ letters, we have $\sigma$ choices for the modus. The inner part of each arch and the rest need to be $(k-\iota)$-universal over $\Sigma\backslash\{\ta\}$ if $\ta$ is the {\em modus-letter}. This allows us to compute the number of shortest words having a minimal number of absent subsequences, of fixed universality.

\begin{proposition}\label{prop:most-number-of-subseq}
For $k,\iota\in\N$, there
exist $\sigma((\sigma-1)!)^{(\iota+1)(k-\iota)}$ different words $w\in \Sigma^*$ of minimal length $ |w| = (\iota+1)\cdot(\sigma -1)\cdot(k-\iota) + \iota$, with universality index $\iota$, each having a maximum number $ \sum_{j \in [0, \iota(w)]} \binom{k}{j} (\sigma-1)^{k - j}$ of subsequences of length $k$. \looseness=-1
\end{proposition}
\ifpaper
\else
\begin{proof}
	Let $\ta\in\Sigma$ with $\m(w)=\ta^{\iota}$. We have $\sigma$ choices for $\ta$. For each such choice,
	we have $\iota$-many inner parts and each inner part is the $(k-\iota)$ power of $w_{\Sigma\setminus\{a\}}$.
	Thus, we have $\iota(k-\iota)$ times the option to choose a permutation over $\Sigma\backslash\{\ta\}$. Since in fact the rest behaves like an arch in the case of the minimality of absent subsequences, we have to count one additional arch.\qed
\end{proof}
\fi

This concludes all results regarding the lower bound on the number of absent subsequences among words with a common universality index.


	\section{Upper Bound for the Number of Absent Subsequences}\label{sec:upbound}
In this section, we aim to identify words, over an alphabet $\Sigma$ of size $\sigma$, with fixed universality index $\iota$, which have a maximum possible number of absent subsequences (of length $k$). We do so by constructing a $\iota$-universal word $w_{\min}$ that, in fact, has the minimum number of existing subsequences for all possible lengths $k>\iota$.
To this end, note that it is not enough to show, for some universality index $\iota$, that $w_{\min}$ has the fewest subsequences for some fixed length $k > \iota$ to get that this holds for all other lengths $k' > \iota$. One can verify that this does not hold in general. 
\ifpaper
\else
\begin{example}\label{counterexample}
    Let $w = \mathtt{(aabbccd) \cdot d}$ and $w' = \mathtt{(abcd) \cdot ccdc}$ with $\iota(w) = 1 = \iota(w')$. We observe
    \begin{itemize}
    \item $|\ScatFact_2(w)| = |\{\ta\ta, \ta\tb, \ta\tc, \ta\td, \tb\tb, \tb\tc, \tb\td, \tc\tc, \tc\td, \td\td\}| = 10$
    and
    \item $|\ScatFact_2(w')| = |\{\ta\tb, \ta\tc, \ta\td, \tb\tc, \tb\td, \tc\tc, \tc\td, \td\tc, \td\td\}| = 9$.
    \end{itemize}
    On the other hand, 
    \begin{itemize}
    \item $|\SubWords_3(w)|=|\{\ta\td\td,\ta\tc\td,\ta\tc\tc,
    \ta\tb\td,\ta\tb\tc,\ta\tb\tb,\ta\ta\td,\ta\ta\tc,\ta\ta\tb$, $\tc\td\td,\tc\tc\td,\tb\td\td$, $\tb\tc\td,\tb\tc\tc,\tb\tb\td, \tb\tb\tc\}| = 16$ and
    \item
    $|\ScatFact_3(w')| = |\{\ta\tb\tc,\ta\tb\td,\ta\tc\td,\ta\tc\tc,\ta\td\tc,\ta\td\td,\tb\tc\td,\tb\tc\tc,\tb\td\tc,\tb\td\td,$ 
    $\tc\td\tc,\tc\td\td$, $\tc\tc\td,\tc\tc\tc,\td\tc\tc,\td\tc\td,\td\td\tc\}| = 17$. 
    \end{itemize}
    Thus, $|\ScatFact_2(w)| > |\ScatFact_2(w')$ but $|\ScatFact_3(w)| < |\ScatFact_3(w')|$.
\end{example}

\fi
Hence, to show our results, we need to have a more general approach and not focus on one specific value $k$, but consider all $k\in[\iota(w)+1,|w|]$.\looseness=-1

First, we define the function $\nextAlphPos_w(i, s)$ taking, as input, some alphabet size $s \in [\sigma]$ and index $i \in [\vert w \vert]$. The function returns the leftmost position $j$ following after $i$ in $w$ such that $w[i+1:j]$ contains $s$ distinct letters or $-1$ if no such position exists, i.e., if $\alphabet(w[i + 1, \vert w \vert]) < s$. 
 \begin{definition}
     Let $w\in\Sigma^*$. We define $\nextAlphPos_w : [|w|]\times[|\Sigma|] \rightarrow [|w|]\cup\{-1\}$ by
     $$ \nextAlphPos_w(i,s) = 
     \begin{cases}  
         j   &, \text{ if } |\alphabet(w[i+1:j])| = s \text{ and } |\alphabet(w[i+1:j-1])| < s, \\
         -1  &, \text{ otherwise.}
     \end{cases}$$
 \end{definition}
 
As an important concept for this section, as well as for Section \ref{sec:enum}, we recall the {\em greedy embedding} of a subsequence $v \in \Sigma^k$ occurring in a word $w$, of length $n$. For $i\in [k]$, we define $j_i$ as the leftmost position of $w$, such that $w[1:j_i]$ contains the subsequence $v[1:i]$. Then, $w[j_1]\cdots w[j_k]=v$ and the sequence $j_1,\ldots,j_k$ is the {\em greedy embedding} of $v$ in $w$, and can be computed in linear time (see, e.g., \cite{DBLP:journals/fuin/KoscheKMS22}). Worth noting, this sequence $j_1,\ldots,j_k$ of indices satisfies $j_i \leq j_i'$, for any $i \in [1, k]$ and sequence $1 \leq j_1' < j_2' < \dots < j_k' \leq \vert w \vert$ with $w[j_1'] w[j_2'] \dots w[j_k'] = v$.


\begin{lemma}
    \label{lemma:next-function-property-scatfact-implies-sequence}
    For $w \in \Sigma^*$, $k \in [\vert w \vert]$, and $u \in \ScatFact_k(w)$, let $ j_1 < \dots < j_k $ be the greedy embedding of $u$ in $w$. Then $\nextAlphPos_w(j_i, \vert \alphabet(w[j_i + 1, j_{i + 1}]) \vert) = j_{i+1}$.
\end{lemma}
\ifpaper
\else
\begin{proof}
	Let $u\in\ScatFact_k(w)$ and assume $e = (j_1,...,j_k)$ to be the greedy embedding of $u$ in $w$. Let $i\in[k-1]$. Assume $s = |\al(w[j_i : j_{i+1}])|$. We know that $w[j_{i+1}]\notin\al(w[j_i : j_{i+1}-1])$, as otherwise $e$ would not be a greedy embedding. Hence, $|\al(w[j_i : j_{i+1}-1])| < s$. By the definition of $\nextAlphPos_w$, we obtain $\nextAlphPos_w(j_i,|\al(w[j_i : j_{i+1}])|) = j_{i+1}$.
\end{proof}

\fi

Using the properties of greedy embeddings, we give the specific construction of words with the minimum number of subsequences, for a given universality index and subsequence length. In particular, we show that all words $w\in\Sigma^*$ that satisfy the structure $\arch_1(w) = \ta_1 \cdots \ta_\sigma$ and $\arch_i(w) = \arch_{i-1}(w)^R$ have the minimum number of subsequences amongst all words with universality index $\iota(w)$ over that alphabet $\Sigma = \{ \ta_1, \ta_2, \dots , \ta_\sigma\}$. We assume, for the remainder of this section, that $w_{\min}$ has this structure, with $\arch_1(w_{\min}) = \ta_1 \cdots \ta_\sigma$, noting that the same construction can be applied for any permutation of $\Sigma$. 
\ifpaper
\else
Example \ref{example:minimal-existing-word-example} presents a visualisation of the structure for a four-letter alphabet.\looseness=-1
\begin{example}\label{example:minimal-existing-word-example}
    Let $\Sigma := \{\ta,\tb,\tc,\td\}, \iota(w) = 5$. Then
    $w =( \ta\tb\tc\td) \cdot (\td\tc\tb\ta)\cdot (\ta\tb\tc\td)\cdot (\td\tc\tb\ta)\cdot (\ta\tb\tc\td)$
    has the minimum number of subsequences of any $w' \in \Sigma^*$ with $\iota(w') =5$.
    Notice that each arch is the reversed version of its predecessor and that every arch is of minimal possible length.\looseness=-1
\end{example}
\fi
Kosche et al. in \cite{DBLP:journals/fuin/KoscheKMS22} show that $w_{\min}$ has the fewest subsequences of length $\iota(w_{\min}) + 1$ amongst all words over $\alphabet(w_{\min})$ with universality index $\iota(w_{\min})$. We generalise this by showing that $w_{\min}$ has for all $k$, the fewest subsequences of length $k$, amongst the set of words $\{v \in \Sigma^* \vert \iota(v) = \iota(w_{\min}) \}$, noting that for any $k \leq \iota(w_{\min})$, all words in this set contain every word in $\Sigma^k$ as a subsequence.
%

\begin{lemma}\label{lemma:min-existent-sf-word-size}
    Let $w\in\Sigma^*$. If $|w| > \iota(w)\cdot \sigma$, then there exists a word $w'\in\Sigma^*$ with
    $\iota(w') = \iota(w)$ and $|w'| = \iota(w)\cdot \sigma$ such that, for all $k\in[|w|]$, we have $\vert \ScatFact_k(w') \vert \leq \vert \ScatFact_k(w) \vert$.
\end{lemma}
\ifpaper
\else
\begin{proof}
	Let $w$ satisfy $\vert w \vert > \iota(w) \cdot \sigma$. Now, there exists some index $i \in [\vert w \vert]$ such that $\iota(w[1:i - 1] w[i + 1:\vert w \vert]) = \iota(w)$. Further, given $u \in \ScatFact(w[1:i - 1] w[i + 1: \vert w \vert])$, $u \in \ScatFact(w)$ by construction. Thus $\ScatFact(w[1:i - 1] w[i + 1:\vert w \vert]) \subseteq \ScatFact(w)$, and by extension $\vert \ScatFact_k(w[1:i - 1] w[i + 1:\vert w \vert]) \vert \leq \vert \ScatFact_k(w) \vert$ for all $k\in[|w|]$, completing the proof.\qed
\end{proof}
\fi

Lemma~\ref{lemma:min-existent-sf-word-size} allows us to only consider words $w\in\Sigma^*$ with $|w| = \iota(w)\cdot \sigma$, which have the property that every arch is a permutation of the alphabet $\Sigma$.
The next two lemmata give upper bounds on possible values of $\nextAlphPos_w$ when $|w| = \iota(w)\cdot\sigma$. Further, we show that the word $w_{\min}$ matches these bounds.\looseness=-1

\begin{lemma}\label{lemma:shortest-w-next-function-upper-bound}
    Let $w\in\Sigma^{\iota(w) \sigma}$. Then, for all $s \in [\sigma]$ and $j\in[|w|]$,
    $ \nextAlphPos_w(j,s) \leq \left\lfloor\frac{j}{\sigma}\right\rfloor\cdot\sigma + \sigma + s. $
\end{lemma}
\ifpaper
\else
\begin{proof}
	Let $w\in\Sigma^{k\cdot\iota}$ be a $\iota$-universal word of length $k\cdot\iota$. We now show a general upper bound for $\nextAlphPos_w(j,s)$, for some position $j\in[|w|]$ and a number of distinct letters $s\in[|\sigma|]$. This means that we want to show the upper bound of the first position of $w$ where $s$ distinct letters occur after position $j$ in $w$. Notice that, due to the length and universality of $w$, each letter occurs in each arch exactly once. Also notice that $r(w)=\varepsilon$. Hence, using the position $j$ and the size of the alphabet $\sigma$, we can exactly define the position $r$ inside the arch in which $w[j]$ is located.
	In particular, let $r = j\bmod\sigma$ if $j\bmod\sigma\neq0$, otherwise let $r = \sigma$.
	Note, first, that if $s \geq \sigma - r$, then we have to jump to the next arch to read $s$ new letters after position $j$. Assume $w[j]$ is in the $i^{\mathtt{th}}$ arch of $w$.
	In the worst case, $\al(\ar_i(w)[r+1 : \sigma]) = \al(\ar_{i+1}(w)[1 : \sigma-r])$, i.e., the letters occurring in the suffix of $\ar_i(w)$ occur as prefix letters in $\ar_{i+1}(w)$. Then, we can assume that $\nextAlphPos_w(j,s)$ returns the position after which $s$ distinct letters have been read in $\ar_{i+1}(w)$. Due to the characterization of $w$, we know that this is position $s$ of $\ar_{i+1}(w)$, thus we have
	$$ \nextAlphPos_w(j,s) \leq \left\lfloor\frac{j}{\sigma}\right\rfloor\cdot\sigma + \sigma + s, $$
	where $\left\lfloor\frac{j}{\sigma}\right\rfloor\cdot\sigma + \sigma$ refers to the starting position of $\ar_{i+1}(w)$.

	Now, assume $s < \sigma-r$. Then, we have $|\ar_i(w)[r+1 : \sigma]| > s$. We know that all letters occur exactly once in each arch. Hence, we know that the factor $\ar_i(w)[r+1 : r+s]$ indeed exists and has an alphabet size of $|\al(\ar_i(w)[r+1 : r+s])| = s$. Hence, $\nextAlphPos_w(j,s) = j+s \leq \left\lfloor\frac{j}{\sigma}\right\rfloor\cdot\sigma + \sigma + s$.

	Finally, if $j > |w| - \sigma$, and thus $j$ is in the last arch of $w$, either $r \leq \sigma - s$, and thus $\vert \alphabet(w[j + 1:j + s]) \vert = s$, so the previous case holds, or $r > \sigma - s$, in which case $\vert \alphabet(w[j + 1:\vert w \vert]) \vert < s$ and thus $\nextAlphPos_w(j,s) = -1$, satisfying the statement.\qed
	
	
\end{proof}
\fi


\begin{lemma}\label{lemma:shortest-w-next-function-upper-bound-same-arch}
    Let $w\in\Sigma^{\iota(w) \sigma}$.
    For $s\in[\sigma]$ and $j\in[|w|]$, if either $j \bmod \sigma = 0$ or $j \bmod \sigma \neq 0$ and $s \leq \sigma-(j\bmod\sigma)$, then
    $ \nextAlphPos_w(j,s) = j+s$ or $ \nextAlphPos_w(j,s) = -1$.\looseness=-1
\end{lemma}
\ifpaper
\else
\begin{proof}
	This result follows analogously to the second and third cases of the proof of Lemma~\ref{lemma:shortest-w-next-function-upper-bound}. In particular, if $|w| - j < \sigma$ and $j \bmod \sigma = 0$, then $j = |w|$, so $\nextAlphPos_w(j,s) = -1$ for all $s\in[\sigma]$. 
	Otherwise, assume first that $j \bmod \sigma = 0$. Then $w[j + 1: j + s]$ is the prefix of an arch of length $\sigma$ containing $s$ unique symbols, and thus $\vert \alphabet(w[j + 1 : s]) \vert = s$. Otherwise, $w[j + 1: j + s]$ is a factor of some arch in $w$, and hence must contain $s$ unique symbols, giving the statement.\qed
\end{proof}
\fi

\begin{lemma}\label{lemma:next-function-min-word-properties}
    For the word $w_{\min}$, $s\in[\sigma]$, $j\in[|w|]$, and $r = j\bmod\sigma$, we have
    $$ \nextAlphPos_{w_{\min}}(j, s) = 
    \begin{cases}  
        j+s     & \text{, if } r = 0\text{ and } |w|-j \geq \sigma \text{, or } s \leq \sigma-r, \\
        \left\lfloor\frac{j}{\sigma}\right\rfloor\cdot\sigma + \sigma + s    & \text{, if } r \neq 0, \left\lfloor\frac{j}{\sigma}\right\rfloor+1 < \iota(w_{\min}), \text{ and } s > \sigma-r,\\
        -1 & \text{, otherwise.}
    \end{cases}$$
\end{lemma}
\ifpaper
\else
\begin{proof}
	Case (1) follows directly from Lemma \ref{lemma:shortest-w-next-function-upper-bound-same-arch}.
	Case (2) can be derived from the same arguments as in Lemma \ref{lemma:shortest-w-next-function-upper-bound},  noting that we have $\alphabet(\arch_i(w_{\min})[\sigma - x + 1: \sigma]) = \alphabet(\arch_{i + 1}(w_{\min})[1 : x])$, for every $x \in [\sigma], i \in [\iota(w_{\min})-1]$. As $s > \sigma-r$, the suffix of the arch $\arch_{x}(w_{\min})$ in which $j$ is located, starting at $j+1$ and ending with the end of that arch, has length less than $s$. Hence, the shortest factor of $w_{\min}$ starting at $j + 1$ containing $s$ symbols has the prefix of length $s$ of arch $\arch_{x + 1}(w_{\min})$ as its suffix, ending at position $\left\lfloor\frac{j}{\sigma}\right\rfloor\cdot\sigma + \sigma + s$. Finally, in case (3), $j + s > \iota(w_{\min}) \sigma$, and thus $\vert \alphabet(w_{\min}[j + 1 : \vert w_{\min} \vert) \vert < s$. By extension, $\nextAlphPos_{w_{\min}}(j, s) = -1$, completing the proof.\qed
\end{proof}
\fi

We can now construct an injective mapping of subsequences from words structured like $w_{\min}$ to subsequences of any other word $w'\in\Sigma^{\iota(w_{min})\sigma}$ with the same universality index and the same size. The function is defined by mapping a subsequence $u\in\ScatFact(w_{min})$, with the greedy embedding $(j_1,...,j_{k})$, to the subsequence $v\in\ScatFact(w')$, with the greedy embedding $(j_1',...,j_{k}')$, that can be obtained by continued application of the $\nextAlphPos$ function using the same sequence of numbers $(s_1,...,s_n)\in[\sigma]^k$, starting in both cases in position $0$. In particular, if we have $\nextAlphPos_{w_{min}}(0,s_1) = j_1$, $\nextAlphPos_{w_{min}}(j_i,s_{i+1}) = j_{i+1}$, then we define $(j_1',...,j_{k}')$ such that $\nextAlphPos_{w'}(0,s_1) = j_1'$, and $\nextAlphPos_{w'}(j'_i,s_{i+1}) = j_{i+1}'$, for $i\in[k-1]$. Following the previous results, this mapping is total and injective.



\begin{proposition}\label{proposition:min-sf-word-injective-function}
    For $w_{\min}$ and each $w'\in\Sigma^{*}$ with
    $\iota(w_{\min}) = \iota(w')$ and $|w_{\min}| = |w'|$, there exists an injective function
    $f : \ScatFact_k(w_{\min}) \rightarrow \ScatFact_k(w')$.
\end{proposition}
\ifpaper
\begin{proof}
	Consider $w_{\min}$ as defined before and let $w'\in\Sigma^*$ with $\iota(w_{\min}) = \iota(w')$
	and $|w_{\min}| = |w'|$. Let $k\in[|w_{\min}|]$.
	Let $f : \ScatFact_k(w_{\min}) \rightarrow \ScatFact_k(w')$ and define $f(u) = v$
	if there exist $j_1,\ldots,j_k\in[|w_{\min}|]$, $j_1',\ldots,j_k'\in[|w'|]$, and $s_1,\ldots,s_{k-1}\in[\sigma]$ such that $u = w_{\min}[j_1]\cdots w_{\min}[j_k], v = w'[j_1'] \cdots w'[j_k'], j_1 = j_1',|w_{\min}[1:j_1-1]|_{w_{\min}[j_1]} = 0, \nextAlphPos_w(j_i,s_i) = j_{i+1}$, and $\nextAlphPos_w(j_i',s_i) = j_{i+1}'$.
	We now show that $f$ is an injective function.
	Let $u\in\ScatFact_k(w_{\min})$.
	By Lemma \ref{lemma:next-function-property-scatfact-implies-sequence} we know that there exist $j_1,\ldots,j_k\in[|w_{\min}|]$ and
	$s_1,\ldots,s_{k-1}\in[\sigma]$ such that $|w_{\min}[1:j_1-1]|_{w_{\min}[j_1]} = 0$ and $\nextAlphPos_{w_{\min}}(j_i,s_i) = j_{i+1}$
	for all $i\in[k-1]$.
	By Lemmas~\ref{lemma:shortest-w-next-function-upper-bound}, ~\ref{lemma:shortest-w-next-function-upper-bound-same-arch}
	and ~\ref{lemma:next-function-min-word-properties} we know that for all $i\in[k-1]$, $\nextAlphPos_{w_{\min}}(j_i,s_i)$
	witnesses its upper bound over all words $w'\in\Sigma^{\vert w_{\min} \vert}$ with $\iota(w_{\min}) = \iota(w')$. Therefore, $\nextAlphPos_{w_{\min}}(j_i,s_i) \geq \nextAlphPos_{w'}(j_i',s_i)$, for every $i \in [k]$, and thus the word constructed by $f$ must be a subsequence of $w'$. To show that $f(u_1) \neq f(u_2)$, for any pair $u_1, u_2 \in \ScatFact_k(w_{\min})$ where $u_1 \neq u_2$, let $i$ be the index such that $u_1[1: i - 1] = u_2[1 : i - 1]$, and $u_1[i] \neq u_2[i]$. Further, let $j_1, \dots, j_k$ be the leftmost set of indices such that $u_1 = w_{\min}[j_1] \dots w_{\min}[j_k]$, and $\ell_1, \dots, \ell_k$ be the leftmost set of indices such that $u_2 = w_{\min}[\ell_1] \dots w_{\min}[\ell_k]$. First, as $u_1[1:i - 1] = u_2[1:i - 1]$, we have by definition $j_1, \dots, j_{i - 1} = \ell_1, \dots, \ell_{i - 1}$. Further, as $u_1[i] \neq u_2[i]$, $j_i \neq \ell_i$, and thus $\vert \alphabet(w_{\min}[j_{i - 1} + 1 : j_{i}]) \vert \neq \vert \alphabet(w_{\min}[\ell_{i - 1} + 1 : \ell_i]) \vert$. Now, let $j_1', \dots, j_k'$ be the leftmost set of indices embedding $f(u_1)$ into $w'$, and let $\ell_1', \dots, \ell_k'$ be the leftmost set of indices embedding $f(u_2)$ into $w'$. Again, we note that $j_1', \dots, j_{i - 1}' = \ell_1', \dots, \ell_{i - 1}'$. Now, we have $j_i' = \nextAlphPos(j_{i - 1}', \vert \alphabet(w_{\min}[j_{i - 1} + 1 : j_{i}]) \vert)$ and $\ell_i' =  \nextAlphPos_{w_{\min}}(j_{i - 1}', \vert \alphabet(w_{\min}[\ell_{i - 1} + 1 : \ell_i]) \vert)$. As $\vert \alphabet(w_{\min}[j_{i - 1} + 1 : j_{i}]) \vert \neq \vert \alphabet(w_{\min}[\ell_{i - 1} + 1 : \ell_i]) \vert$, $w'[j'_i] \neq w'[\ell_i']$, and thus $f(u_1)[i] \neq f(u_2)[i]$. Due to the shape of $w'$, we have $f(u_1) \neq f(u_2)$.\qed
\end{proof}
\else

\fi

Using this construction, we can conclude the main result of this section.

\begin{theorem}\label{theorem:min-sf-word}
    Let $\Sigma := \{\ta_1,\ldots,\ta_\sigma\}$ for some $\sigma\in\N$.
    Let $w\in\Sigma^*$ such that w.l.o.g. $\arch_1(w) = \ta_1\cdots\ta_\sigma$ and $\arch_i(w) = \arch_{i-1}(w)^R$ for all $i\in[2,\iota(w)]$. Then, for all $w'\in\Sigma^*$ with $\iota(w) = \iota(w')$, we have $|\ScatFact_k(w)| \leq |\ScatFact_k(w')|$.
\end{theorem}
\ifpaper
\else
\begin{proof}
	Let $w$ be given as in the statement and let $k\in[|w|]$.
	Let $w'\in\Sigma^*$ be any other word with $\iota(w) = \iota(w')$.
	By construction of $w'$, $|w'| \geq \iota(w) \sigma$.
	If $|w'| > \iota(w')\cdot\sigma$ holds, we know by Lemma \ref{lemma:min-existent-sf-word-size} there exists
	a word $w''\in\Sigma^*$ with $\iota(w'') = \iota(w')$, $|w''| = \iota(w)\cdot\sigma$, and $|\ScatFact_k(w'')| \leq |\ScatFact_k(w')|$.
	Now, by Proposition \ref{proposition:min-sf-word-injective-function}, there exists an injective function
	$f : \ScatFact_k(w) \rightarrow \ScatFact_k(w'')$.
	By that, we know that $|\ScatFact_k(w)| \leq |\ScatFact_k(w'')|$ and as $|\ScatFact_k(w'')| \leq |\ScatFact_k(w')|$, we also
	have $|\ScatFact_k(w)| \leq |\ScatFact_k(w')|$.\qed
\end{proof}
\fi

There exist cases in which only the proposed structure of $w_{min}$ results in a maximal number of absent subsequences in all cases $k\in[\iota(w)+1,|w|-1]$, in particular when $\iota(w) = 2$ and $\sigma = 3$: see \Cref{tableminimalexisting}.
\begin{table}[h!]
	\begin{tabular}{c|c|c|c|c|c|c|}
		& $\mathtt{(abc)\cdot (abc)}$ & $\mathtt{(abc) \cdot (acb)}$ & $\mathtt{(abc) \cdot (bac)}$ & $\mathtt{(abc) \cdot (bca)}$ & $\mathtt{(abc) \cdot (cab)}$ & $\mathtt{(abc) \cdot (cba)}$ \\
		\hline
		$k=3$ & 17 & 16 & 16 & 14 & 14 & 13\\
		$k=4$ & 15 & 14 & 14 & 13 & 11 & 11\\
		$k=5$ & 6 & 6 & 6 & 6 & 5 & 5
	\end{tabular}
	\caption{Numbers of existing subsequences per perfect $2$-universal word for $k \in [3,5]$.}\label{tableminimalexisting}
\end{table}

Note that the previous results do not exclude the existence of $\iota$-universal words $w\in\Sigma^{\iota\sigma}$ that do not follow the structure of $w_{min}$ but have a minimal number of existent subsequences over all $k \in [\iota(w)+1,|w|-1]$. 

	
	\section{Enumeration Algorithms}\label{sec:enum}
	So far, we focused on the combinatorial properties of the set of subsequences of length $k$ of a word $w$ (and their implications w.r.t. the set of missing subsequences). In this section, we consider the algorithmic problem of enumerating the elements of this set. 
In particular, we propose optimal  enumeration algorithms for the set of subsequences of length $k$ of a given word. 
Enumeration algorithms with output-linear delay for this set of subsequences can be obtained (as explained below) via subsequence-automata \cite{CrochemoreMT03}, while counting the length-$k$ subsequences of $w$ can be done in $O(k|w|)$ time by a dynamic programming algorithm \cite{ELZINGA2008394} (which can also be extended to an enumeration algorithm). However, such enumeration algorithms are not optimal, and our results improve them significantly. We are not aware of any other results related to the efficient (incremental) enumeration of the subsequences of a word; recent related works approach, based on totally different ideas, the problems of designing efficient algorithms for the enumeration of the longest common subsequences of two strings (see \cite{Sakai25,Sakai25-2} and the references therein) and for the enumeration of the maximal common subsequences of two strings (see \cite{ConteGPU22,BuzzegaEtAl} and the references therein). 
First, we are going to present an informal overview of the procedure.

Firstly we give a high-level overview our algorithmic results; the reader not interested in the technical details of our algorithms can read the next subsection only, and skip Section \ref{sec:enum_tech_details}. The reader interested in these details can also directly go to the respective subsection.

\subsection{Overview}

Our results are based on two key ideas. The first is related to the notion of {\em greedy embedding} of a subsequence $v$, of length $k$, occurring in a word $w$, of length $n$, introduced in Section \ref{sec:upbound} in connection to Lemma \ref{lemma:next-function-property-scatfact-implies-sequence}. 
In the context of this section, we note that two distinct greedy embeddings correspond to two distinct subsequences, and every subsequence has a greedy embedding, we conclude that it is enough to enumerate the sets of greedy embeddings of subsequences of length $k$. \looseness=-1

Secondly, we recall that one can construct in time $O(n\sigma)$ a deterministic finite automaton $A = (Q,\Sigma,q_0, F, \delta)$ with $L(A) = \{v | v$ is a subsequence of $w\}$; see, e.g., \cite{CrochemoreMT03}. In this automaton, we have the states $Q = \{0,\ldots,n,n+1\}$, $q_0 = 0$, $F=Q\setminus \{0,n+1\}$. Moreover, we define the transition function $\delta:Q\times\Sigma\rightarrow Q$ as follows. For $i\in [n]$ and $a\in \Sigma$, we set $\delta(i,a) =  \min (\{n+1\}\cup \{j\mid i<j\leq n, w[j]=a\}) $; further, $\delta(n+1,a)=n+1$ for all $a\in \Sigma $ (i.e., $n+1$ is an error-state). Importantly, for a subsequence $v$, of length $m$, occurring in $w$, the automaton $A$ will accept $v$ while going through the sequence of states $0, i_1, \ldots, i_m$, where $i_1, \ldots, i_m$ is the greedy embedding of $v$ in $w$. 

Thus, we may enumerate all subsequences of length $k$ of $w$ (more precisely, their greedy embeddings) after an $O(n\sigma)$ time preprocessing, with $O(k)$ delay, by way of a depth-first traversal over the length $k$ paths in the graph corresponding to the automaton $A$.
This can be improved, as long as we are interested in an incremental enumeration: we simply produce the first output (the subsequence $w[1:k]$, represented, e.g., as pair $(i,k)$), and then, in subsequent steps of our enumeration, we output a short sketch of the edits needed to be performed to obtain the current subsequence from the previous one. In particular, if these edits are "remove suffix of length $\ell$, add $w[a+1:a+\ell]$ as a new suffix to the current subsequence" (where in the last part, $w[a:b]$ is simply represented as the pair of positions $a,b$ of $w$), then we can obtain, using the techniques from \cite{AdamsonGM24}, an enumeration algorithm with $O(n\sigma)$ time preprocessing and $O(1)$ delay. This algorithm allows, at any moment, to output explicitly the currently enumerated subsequence.\looseness=-1

We improve these results, and show that, both in the case of an explicit enumeration algorithm and in that of an incremental enumeration algorithm, the preprocessing time can be reduced to $O(n)$ and, hence, we obtain optimal explicit and, respectively, incremental enumeration algorithms. \looseness=-1

The main obstacle in achieving $O(n)$ preprocessing time is the explicit construction of the DFA $A$, of size $O(n\sigma)$. However, when working with $A$, we only needed the states $Q$, and for each $i\in Q\setminus \{n,n+1\}$ the set of its relevant successor-states, $L_i=\{j\neq n+1 \mid \exists a\in \Sigma, \delta[i,a]=j\}$. The key observation is that we can identify the elements of $L_i$, for $i<n$, on the fly, and we do not need to precompute these sets. However, we need some additional data structures: the array $\prevArray[\cdot]$, of size $n$, where $\prevArray[i]= \max \lbrace j \mid w[j] = w[i], j < i \rbrace$, and $\prevArray[i] = 0$ if $w[i] $ does not occur in $w[1:i-1]$; and range minimum query data structures for $\prevArray[\cdot]$ (denoted $RMQ_{\prevArray}$). Now, back to the computation of $L_i$, we always have $i+1\in L_i$ (for $i<n$). All the remaining elements of $L_i$, if any, are in the range $[i+2:n]$; so we define a list $I_i$ containing only this range. Then we iterate the following process (which simulates the traversal of the list $L_i$, in the case when it was precomputed), but on the fly, each time we need to identify a new state-successor of $i$. Until $I_i$ becomes empty, we extract the range $[a:b]$ from $I_i$ and set $j=RMQ_{\prevArray}[a,b]$. If $\prevArray[j]\leq i$, then $j\in L_i$ holds, and further elements of $L_i$ can be potentially found in the ranges $[a,j-1]$ and $[j+1:b]$, so we add them to $I_i$. If $\prevArray[j]> i$, then there is no element of $L_i$ in the range $[a:b]$, as all $\prevArray[\ell]>i$ for all $\ell\in [a:b]$. This process allows us to identify the successors of a state in $A$ in $O(1)$ time per successor, with $O(n)$ preprocessing, and it immediately leads to an implementation of the previous enumeration algorithm with $O(n)$ preprocessing time and $O(k)$ delay. This is optimal when we want the outputs explicitly. 

All the algorithms based on the exploration of $A$ maintain a stack which contains the positions $i_1,\ldots,i_k$ of the greedy embedding of the currently enumerated subsequence. So, to get to $O(1)$ delay (for incremental enumeration, as described above), the novel idea is to maintain the greedy embedding of the currently enumerated subsequence in a more succinct way, which allows us to identify quicker the next subsequence to be enumerated, and still do the enumeration in the same order as in the previous algorithms. This succinct representation is actually relatively simple: as contiguous factors of the currently enumerated subsequence $v$ can also be embedded, in the greedy embedding of $v$ in $w$, into contiguous factors of $w$, we will not maintain this embedding of $v$ as a sequence of separate positions of $w$, but as a sequence of maximal contiguous factors of $w$ (that is, we represent compactly, as a range $[i:j]$, the infixes of the greedy embedding which correspond to the factor $w[i:j]$). This representation can be implemented in a way that allows us to identify in $O(1)$ the longest prefix (i.e., the rightmost position $i_j$) of the greedy embedding $i_1,\ldots, i_k$ of $v$ in $w$ such that the sequence of positions $i_1,\ldots, i_j$ can be extended to obtain a subsequence not enumerated before. A key idea here is that this extension can also be done in $O(1)$: we perform this extension with a sequence $i'_{j+1},i'_{j+1}+1,\ldots, i'_{j+1}+k-1 $ (represented succinctly as the range $[i'_{j+1}:i'_{j+1}+k-1 ]$, corresponding to the factor $w[i'_{j+1}: i'_{j+1}+k-1]$ of $w$). We choose $i'_{j+1}$ exactly in the same way as in the RMQ-based algorithm, overviewed above, which will guarantee that we did not enumerate any subsequence whose greedy embedding starts with $i_1,\ldots, i_j, i'_{j+1}$ until now. Moreover, our algorithm enumerates then all subsequences whose greedy embedding starts with $i_1,\ldots, i_j, i'_{j+1}$, and then moves on and finds a new position $i''_{j+1}$ and enumerate all subsequences whose greedy embedding starts with $i_1,\ldots, i_j, i''_{j+1}$. 
The result is an enumeration algorithm for the set of subsequences of length $k$ of $w$, with $O(n)$ preprocessing time and $O(1)$ delay. The currently enumerated subsequence can be also output, in $O(k)$ time. \looseness=-1

\subsection{Technical Details}\label{sec:enum_tech_details}

The computational model we use to state our algorithms is the standard unit-cost word RAM with logarithmic word-size $\omega$ (meaning that each memory word can hold $\omega$ bits). It is assumed that this model allows processing inputs of size $n$, where $\omega \geq \log n$; in other words, the size $n$ of the data never exceeds (but, in the worst case, is equal to) $2^\omega$. Intuitively, the size of the memory word is determined by the processor, and larger inputs require a stronger processor (which can, of course, deal with much smaller inputs as well). Indirect addressing and basic arithmetical operations on such memory words are assumed to work in constant time. Note that numbers with $\ell$ bits are represented in $O(\ell/\omega )$ memory words, and working with them takes time proportional to the number of memory words on which they are represented. This is a standard computational model for the analysis of algorithms, defined in \cite{FredmanW90}. Our algorithms have strings and numbers as input, so we follow standard assumptions, Namely, we work with languages over {\em integer alphabets} (see, e.\,g.,~\cite{crochemore}): whenever we are given a string of size $n$ as input, we assume that its alphabet $\Sigma=\{1,2,\ldots,\sigma\}$, with $|\Sigma|=\sigma\leq n$.

\begin{observation}\label{preprocessing}
The following data structures were introduced in \cite{DBLP:journals/fuin/KoscheKMS22,KoscheKMP22}, and we recall them here. Assume $w$ is a word of length $n$. 
\begin{itemize}
\item In \cite{KoscheKMP22} it is shown that we can construct the $n\times \sigma$ matrix $\nextpos[\cdot][\cdot]$, where $\nextpos[i][a] = \min (\{n+1\}\cup \{j \in [i + 1, n] \mid w[j]=a\}) $, for $i\leq n$ and $a\in \Sigma$. Intuitively, $\nextpos[i][a]$ is the first position of the word $w$ where $a$ occurs after position $i$ (or $n+1$ if there is no such occurrence).  
\item In \cite{DBLP:journals/fuin/KoscheKMS22} it is shown that we can construct in $O(n)$ time the array $\prevArray[\cdot]$, of size $n$, where $\prevArray[i]$ stores the rightmost position of $w[i]$ in the word $w[1:i-1]$.
That is, $\prevArray[i] = \max \lbrace j \in [1, i -1] \mid w[j] = w[i]\rbrace$, while we assume $\prevArray[i] = 0$ if $w[i] $ does not occur in $w[1:i-1]$. 
\end{itemize}
\end{observation} 

We recall the {\em greedy embedding} of a subsequence $v \in \Sigma^m$ occurring in a word $w$, of length $n$ (introduced in Section \ref{sec:upbound} in connection to Lemma \ref{lemma:next-function-property-scatfact-implies-sequence}). For $j\in [m]$, we define $i_j$ as the leftmost position of $w$, such that $w[1:i_j]$ contains the subsequence $v[1:j]$. Then, $w[i_1]\cdots w[i_m]=v$. The sequence $i_1,\ldots,i_m$ is the {\em greedy embedding} of $v$ in $w$, and can be computed in linear time (see, e.g., \cite{DBLP:journals/fuin/KoscheKMS22}). 

\begin{lemma}\label{lemmaDFA}
Given a word $w \in \Sigma^*$, with $|w|=n$, $|\Sigma|=\sigma$, and an integer $k \leq n$, we can enumerate all subsequences of length $k$ with $O(n \sigma)$ time preprocessing and $O(k)$ time delay.
\end{lemma}
\begin{proof}
We first build the data structures from Observation \ref{preprocessing}.

We can construct in time $O(n\sigma)$ a deterministic finite automaton $A = (Q,\Sigma,q_0, F, \delta)$ with $L(A) = \{w' | w'$ is a subsequence of $w\}$; see, e.g., \cite{CrochemoreMT03}. We have $Q = \{0,\ldots,n,n+1\}$, $q_0 = 0$, $F=Q\setminus \{0,n+1\}$. We define the transition function $\delta:Q\times\Sigma\rightarrow Q$ as follows. For $i\in [n]$ and $a\in \Sigma$, we set $\delta(i,a) = \nextpos[i,a]$; further, $\delta(n+1,a)=n+1$ for all $a\in \Sigma $ (i.e., $n+1$ is an errror-state).

Recall that, for a subsequence $v$, of length $m$, occurring in $w$, the automaton $A$ will accept $v$, going through the sequence of states $0, i_1, \ldots, i_m$, where $i_1, \ldots, i_m$ is the greedy embedding of $v$ in $w$. 

Given $A$, we can compute, for each state $i\in Q$, a set $L_i$ (implemented as list) of the targets of the $\sigma$ transitions leaving state $i$; more precisely, $L_i = \{\delta(i,\ell)\mid \ell \in \Sigma \}$. Note that $i+1\in L_i$, for all $i\in Q\setminus \{n+1\}$, as $i+1=\delta(i,w[i+1])$.  In the following we will consider the set $L_i$ to be implemented as linear list, and thus will talk of the list $L_i$, rather than of the set $L_i$. This implementation of $L_i$ induces a total order between the elements of this set. For simplicity, we assume that $i+1$ is the first element of the list $L_i$, while the order of the rest of the elements of $L_i$ is arbitrary. 


Once we have the lists $L_i$, we can proceed to give the enumeration algorithm of the length $k$ subsequences of $w$. Intuitively, an ordered tree $T$, with labels on both nodes and edges, is induced by the lists $L_i$. The root of this tree is labelled with $0$, its nodes are labelled with states $i\in Q\setminus \{n+1\}$, and the children of a node whose label is $i$ are nodes whose labels are $j< n+1$ such that $j\in L_i$. The edge connecting a node with label $i$ and a node with label $j$ is labelled with $w[j]$ (and corresponds to the transition $j=\delta(i,w[j])$). The label of a walk in this tree is obtained by concatenating the labels of the edges on that walk. We also assume that the children of a node with label $i$ are ordered according to the order of the nodes of the list $L_i$, with a node labelled by $i+1$ being the first child of the node labelled with $i$. As the distinct subsequences of length $k$ of $w$ correspond bijectively to distinct accepting paths of length $k$ in the automaton $A$ (as each such path is a greedy embedding of a subsequence in $w$, and every two greedy embeddings correspond to distinct subsequences), it is immediate that they also correspond bijectively to the distinct paths of length $k$ in the tree $T$, starting from the root. Our enumeration algorithm simulates a depth-first traversal of this tree, up to the nodes on level $k+1$ (assuming that the root is on level $1$), which outputs the labels of every walk. Importantly, note that at this point there are no two such paths labelled by the same word, as the labels of every two edges leaving a node are distinct. Nevertheless, constructing this tree effectively in a preprocessing phase is too time consuming (and not necessary, as the lists $L_i$ contain already all the information we need, and are, as such, an implicit representation of $T$). We briefly explain in the following how our enumeration algorithm works. 

In our algorithm, we maintain a stack $S$ whose elements are pairs $(i,\ell)$, where $i\in Q\setminus \{n+1\}$ and $\ell$ is a position of the list $L_i$.

Initially, we set $S=(0,1)$. 

We now describe a general step of our algorithm. Let 
\[
S=((0,\ell_0), (i_1, \ell_1), \ldots, (i_t, \ell_t))
\]
(where the top of the stack is its rightmost element $(i_t, \ell_t)$). Following the intuition above, the stack $S$ contains the labels $0,i_1,\ldots,i_t$ of the nodes on the currently explored walk starting with the root, which has length $t$ and label $w[i_1]\cdots w[i_t]$. For each $j\in [0:t]$, $\ell_j$ encodes the fact that our traversal already completely explored the first $\ell_j-2$ children of the $(j+1)^{th}$ node on this path (which has label $i_j$), is now exploring the paths that go through the $(\ell_j-1)^{th}$ child of that node, and will explore, once these paths are completely explored, those going through the $\ell_j^{th}$ child of the respective node. 

If $t=k$, we output $w[i_1]\cdots w[i_t]$ and remove $(i_t, \ell_t)$ from the stack. In this case, we have found a path of length $k$ starting with the root, so we need to backtrack and find another one. We then repeat this general step.  

If $t<k$ and the element on the top of the stack is of the form $(i_t,\ell_t)$, where the element on position $\ell_t$ of $L_i$ is either $(n+1)$ or some $j>n-(k-t)+1$, we pop it and repeat this general step. In this case, there is no continuation of the path encoded by $S=((0,\ell_0), (i_1, \ell_1), \ldots, (i_t, \ell_t))$ which could lead to a path of length $k$, so we need to backtrack.

If $t<k$ and the element on the top of the stack is of the form $(i_t,\ell_t)$ where the element on position $\ell_t$ of $L_i$ is some $j\leq n-(k-t)+1$, we replace $(i_t, \ell_t)$ by $(i_t, \ell_t+1)$ and push $(j,1)$ in $S$; that is, $S$ becomes $((0,\ell_0), (i_1, \ell_1), \ldots, (i_t, \ell_t+1), (j,1))$, and repeat the general step of the algorithm. Basically, we were able to extend the path encoded by the contents of $S$, in this case. Worth noting, it is guaranteed that $S=((0,\ell_0), (i_1, \ell_1), \ldots, (i_t, \ell_t+1), (j,1))$ encodes a path that either has already length $k$ or can be continued to a path of length $k$, namely the one corresponding to the sequence $((0,\ell_0), (i_1, \ell_1), \ldots, (i_t, \ell_t+1), (j,1), (j+1,1),\ldots, (j+(n-(k-t)+1), 1))$. 

Observe that our algorithm simply implements a depth-first traversal of the tree $T$,  succinctly represented by the lists $L_i$. 
From the explanations given above, it is immediate that each path of length $k$, starting in the root of the respective tree $T$, is output once by our algorithm. Firstly, it is clear that for each greedy embedding of a subsequence (so, for each subsequence), there is a path of length $k$ in the automaton, and in the corresponding tree, so the subsequence corresponding to that embedding will be output at some point in our algorithm. To show that each subsequence is output only once, note that the nodes on such a path of the automaton/tree, say $0,i_1,  \ldots, i_t$, also correspond, by their definition, to a greedy embedding of some subsequence of $w$. After outputting the respective subsequence, our algorithm simply finds the longest prefix of this subsequence (that is, the rightmost position $i_j$ in this subsequence), such that $0,i_1,  \ldots, i_j$ can be completed to obtain a greedy embedding of some subsequence, which was not already output (as it differs on position $j+1$ with any other subsequence that shares the same prefix of length $j$, and was already output). Therefore, each subsequence of length $k$ is output exactly once by our algorithm. Finally, the delay between two outputs is $O(k)$: in the worst case, after outputting a subsequence, we have to pop $k$ elements from the stack, and push $k$ elements in it, and then we will again output a subsequence (or finish the enumeration). Thus, the statement holds. 
\end{proof}

%
%
%
\begin{algorithm}
\caption{Enumerate Subsequences of Length $k$}
\begin{algorithmic}[1]
\Require Word $w$ of length $n$, alphabet $\Sigma$ of size $\sigma$, integer $k \leq n$
\Ensure Enumeration of all subsequences of length $k$ with $O(n\sigma)$ preprocessing and $O(k)$ delay

\State Initialize stack $S = [(0,1)]$ \Comment{Start from state 0, first transition}

\While{$S$ is not empty}
    \State $(i_t, \ell_t) \gets$ top element of $S$
    
    \If{$|S| = k$} \Comment{Valid subsequence found}
        \State Output subsequence corresponding to $S$
        \State Pop $(i_t, \ell_t)$ from $S$
        \State \textbf{continue}
    \EndIf

    \If{$\ell_t > |L_{i_t}|$ or $L_{i_t}[\ell_t] = n+1$ or $L_{i_t}[\ell_t] > n - (k - |S|) + 1$}
        \State Pop $(i_t, \ell_t)$ from $S$ \Comment{Backtrack}
    \Else
        \State $j \gets L_{i_t}[\ell_t]$ \Comment{Next state}
        \State Replace $(i_t, \ell_t)$ by $(i_t, \ell_t + 1)$ in $S$
        \State Push $(j, 1)$ onto $S$
    \EndIf
\EndWhile
\end{algorithmic}
\end{algorithm}

\begin{observation}
Following \cite{AdamsonGM24}, we can preprocess the deterministic finite automaton $A$ defined in the proof of Lemma \ref{lemmaDFA} such that we can enumerate succinct representations of all distinct words of length $k$ accepted by $A$ (i.e., subsequences of length $k$ of $w$), with $O(n \sigma)$ time preprocessing and $O(1)$ time delay. 
\end{observation}

More precisely, for an input directed graph, \cite{AdamsonGM24} defines the default edge for each vertex, and uses the induced default walk in graphs (walks consisting of default edges only) to succinctly represent arbitrary walks in the respective graph. Based on this, the walks of length $k$, starting with a given source-vertex of the input graph, are incrementally enumerated with $O(1)$ delay, after a preprocessing linear in the size of the graph. One starts with the default walk of length $k$ starting from the source (it is enough to output the length of this path). The enumeration then runs by always outputting the length of the suffix of the previous path that needs to be removed, and how a new path of length $k$ is then obtained by following a non-default edge (explicitly output) and again a default path (starting from the target-vertex of the non-default edge, and whose length is that of the removed suffix minus one). Crucially, the algorithm of \cite{AdamsonGM24}  maintains a data structure that allows, at each point of the enumeration, the explicit output of the currently enumerated walk, in optimal $O(k)$ time. As the deterministic finite automaton $A$ is a directed graph, the incremental enumeration (as briefly defined above) of the distinct words of length $k$ accepted by $A$ (i.e., the subsequences of length $k$ of $w$) can be done after an $O(n\sigma)$ time preprocessing, with $O(1)$ delay. 

Let us first improve the result of Lemma \ref{lemmaDFA} and given an optimal enumeration algorithm for the subsequences of a word $w$, when we are interested in effectively outputting these subsequences.
\begin{theorem}\label{lemmaRMQ}
Given a word $w \in \Sigma^*$, with $|w|=n$, $|\Sigma|=\sigma$, and an integer $k \leq n$, we can enumerate all subsequences of length $k$ with $O(n)$ time preprocessing and $O(k)$ time delay.
\end{theorem}

\begin{proof}
This algorithm is a more efficient implementation of the algorithm from Lemma \ref{lemmaDFA}. 

The main issue is that we cannot afford to explicitly construct the automaton $A$, and, more precisely, the lists $L_i$ introduced in the previous proof.
As recalled in Observation \ref{preprocessing}, we construct for the word $w$, in a preprocessing phase, the array $\prevArray[\cdot]$. 
Further, we construct data structures $RMQ_{\prevArray}$ allowing us to answer range minimum queries for the array $\prevArray[\cdot]$; this can be done in linear time \cite{BenderF02}. 

We now show how these data structures implicitly encode the sets $L_i$. Note that, since $L_i=\{j\mid \exists a\in \Sigma, \nextpos[i,a]=j\}$, we have that $L'_i=L_i\setminus \{n+1\}=\{j\mid \prevArray[j]\leq i\}$. Now, the key observation is that we can identify the elements $L'_i$, for $i<n$, as follows. We start with $i+1$, which is always in $L'_i$ (for $i<n$), and note that all the remaining elements of $L'_i$, if any, are in the range $[i+2:n]$; so we define a list $I_i$ containing this range. Then we repeat the following process, until $I_i$ is empty. We extract the range $[a:b]$ from $I_i$ and set $j=RMQ_{\prevArray}[a,b]$. If $\prevArray[j]\leq i$, we note that $j\in L'_i$ holds, and that further elements of $L'_i$ can be potentially found in the ranges $[a,j-1]$ and $[j+1:b]$, so we add them to $I_i$. If $\prevArray[j]> i$, then there is no element of $L'_i$ in the range $[a:b]$, as all $\prevArray[\ell]>i$ for all $\ell\in [a:b]$.

We will now rewrite our algorithm from the proof of Lemma \ref{lemmaDFA}, based on this idea.

As before, we maintain a stack $S$ whose elements are pairs $(i,I)$, where $i\in Q\setminus \{n+1\}$ and $I$ is a queue of disjoint ranges $[a:b]$, with $1\leq a\leq b\leq n$. 

Initially, we set $S=(0,I_0)$, where $I$ contains the single range $[1:n-k+1]$. 

We now describe a general step of our algorithm. Let 
\[
S=((0,I_0), (i_1, I_1), \ldots, (i_t, I_t))
\]
(where the top of the stack is its rightmost element $(i_t, \ell_t)$). In this case, the stack $S$ encodes the subsequence $w[i_1]\cdots w[i_t]$, which we try to extend to obtain a subsequence of length $k$. Moreover, for each $j\in [0:t]$, $I_j$ contains ranges $[a:b]$ which contain a position $\ell$ such that $\prevArray[\ell]\leq i_j$ (so an element of $L'_{i_j}$) and none of the subsequences starting with $w[i_1]\cdots w[i_j]w[\ell]$ was output yet. 

If $t=k$, we output $w[i_1]\cdots w[i_t]$ and remove $(i_t, I_t)$ from the stack. In this case, we have found a path of length $k$ starting with the root, so we need to backtrack and find another one. We then repeat this general step.  

If $t<k$ and the element on the top of the stack is of the form $(i_t,\emptyset)$, we pop it and repeat this general step. In this case, there is no continuation of the path encoded by $S=((0,\ell_0), (i_1, \ell_1), \ldots, (i_t, \ell_t))$ which could lead to a path of length $k$, so we need to backtrack.

If $t<k$ and the element on the top of the stack is of the form $(i_t,I_t)$ with $I_t\neq \emptyset$, we proceed as follows. We extract the first range $[a:b]$ of $I_t$. The first case is when $a=i_t+1$. In this case, let $\ell=RMQ_{\prevArray}[a+2:b]$ and, if $\prevArray[\ell]\leq i_t$, add $[a+2:b]$ to $I_t$. Let $i_{t+1}=i_t+1$ and $I_{t+1}$ contain only the range $[a+2:n-(k-t+1)+1]$, and push $(i_{t+1},I_{t+1})$ in the stack $S$. That is, $S$ becomes $((0,I_0), (i_1, I_1), \ldots, (i_t,I_t), (i_{t+1}, I_{t+1})$, and repeat the general step of the algorithm. Basically, we were able to extend the path encoded by the contents of $S$, in this case. The second case is when $a\neq i_t+1$. In this case, let $\ell_1=RMQ_{\prevArray}[a+1:b]$; by construction, we guarantee that $\prevArray[\ell]\leq i_t$. Further let $\ell_2=RMQ_{\prevArray}[a:\ell_1-1]$ and $\ell_3=RMQ_{\prevArray}[\ell_1+1:b]$. If $\prevArray[\ell_2]\leq i_t$, add $[a:\ell_1-1]$ to $I_t$; if $\prevArray[\ell_3]\leq i_t$, add $[\ell_1+1:b]$ to $I_t$.  Let $i_{t+1}=\ell_1$ and $I_{t+1}$ contain only the range $[\ell_1+1:n-(k-t+1)+1]$, and push $(i_{t+1},I_{t+1})$ in the stack $S$. That is, $S$ becomes $((0,I_0), (i_1, I_1), \ldots, (i_t,I_t), (i_{t+1}, I_{t+1})$, and repeat the general step of the algorithm. Basically, we were able to extend the subsequence encoded by the contents of $S$, in this case. It is worth noting that it is guaranteed that $S=((0,I_0), (i_1, I_1), \ldots, (i_t,I_t), (i_{t+1}, I_{t+1})$ encodes a path that either has already length $k$ or can be continued to a path of length $k$, namely the one corresponding to the sequence $w[i_1] \cdots w[i_{t+1}] w[i_{t+1}+1]\cdots w[i_{t+1}+(k-(t+1))]$. Moreover, each range contained in the queue $I_j$, for all $j\leq t+1$, contains at least one position $\ell$ such that $\prevArray [\ell]\leq j$ such that no subsequence of length $k$ starting with $w[i_1] \cdots w[i_{j}]w[\ell]$ was output yet; in other words, each range of the queue $I_j$ contains at least one position $\ell\in L'_{i_j}$ such that no subsequence of length $k$ starting with $w[i_1] \cdots w[i_{j}]w[\ell]$ was output yet. 

It is not hard to see that the above algorithm implements the algorithm of Lemma \ref{lemmaDFA}, where all the elements of $L'_i$ are identified and explored using $RMQ_{\prevArray}$ queries, instead of having them computed and stored from the beginning. Just like before, we explore all possible greedy embeddings of the subsequences of $w$. For each greedy embedding, we output the corresponding subsequence. Moreover, once we output a subsequence (corresponding to some greedy embedding), we again find the longest prefix (of length $j$) of that embedding that can be completed to obtain a new greedy embedding of length $k$, which was not explored. As before, we guarantee that the new embedding is different from all the already enumerated greedy embeddings with which it shares the prefix of length $j$; therefore, it leads to a new subsequence. The main difference is, as said, that instead of having the list $L'_i$ from the beginning, we compute it on the fly, using $RMQ_{\prevArray}$ queries,
The correctness of the algorithm follows from these observations. The delay is also $O(k)$, by the same arguments as in Lemma \ref{lemmaDFA}.\qed
\end{proof}

\begin{algorithm}
\caption{Enumerate Subsequences using RMQ}
\begin{algorithmic}[1]
\Require Word $w \in \Sigma^*$ of length $n$, integer $k \leq n$, array $\prevArray[\cdot]$
\Ensure All subsequences of length $k$
\State Preprocess $\prevArray[\cdot]$ to build $RMQ_{\prevArray}$
\State Initialize stack $S \gets [(0, \{(1, n-k+1)\})]$

\While{$S$ is not empty}
    \State Let $(i_t, I_t)$ be the top of the stack

    \If{$|S|= k$}
        \State Output subsequence represented by $S$
       \If{$(I_t) \neq \emptyset$}
       	\State Pop top element of $(I_t)$
       \Else
       	\State Pop top element from $S$
       \EndIf
        \State \textbf{continue}
    \EndIf

    \If{$I_t$ is empty}
        \State Pop $(i_t, I_t)$ from the stack
        \State \textbf{continue}
    \EndIf

    \State Extract the first interval $[a, b]$ from $I_t$

    \If{$a = i_t + 1$}
        \State $\ell \gets RMQ_{\prevArray}[a+2, b]$
            \If{$\prevArray[\ell] \leq i_t$}
                \State Append $[a+2, b]$ to $I_t$
            \EndIf
        \State Push $(a, \{(a+2, n - (k - t + 1) + 1)\})$ onto $S$
        \State \textbf{continue}
    \EndIf

    \State $\ell \gets RMQ_{\prevArray}[a+1, b]$

    \If{$\prevArray[\ell] > i_t$}
        \State \textbf{continue}
    \EndIf

    \State $\ell_1 \gets RMQ_{\prevArray}[a, \ell - 1]$
    \State $\ell_2 \gets RMQ_{\prevArray}[\ell+1, b]$

    \If{$\prevArray[\ell_1] \leq i_t$}
        \State Append $[a, \ell - 1]$ to $I_t$
    \EndIf

    \If{$\prevArray[\ell_2] \leq i_t$}
        \State Append $[\ell + 1, b]$ to $I_t$
    \EndIf

    \State Push $(\ell, \{(\ell + 1, n - (k - t + 1) + 1)\})$ onto $S$
\EndWhile
\end{algorithmic}
\end{algorithm}

We can now present the main result of this section. As in \cite{AdamsonGM24}, we propose an incremental enumeration of the subsequences of length $k$ of some string $w$. We start with the subsequence $w[1:k]$ (encoded as the pair $(1,k)$), which is the greedy embedding of the word $w[1:k]$ in $w$. Then we enumerate the subsequences of length $k$ of $w$ in the sense that we output the length $\ell$ of the suffix of the previously enumerated subsequence which should be removed, and replaced with some factor $w[a+1:a+\ell]$ (encoded as $(a+1,a+\ell)$). Simultaneously, we maintain a data structure, from which the current subsequence can be effectively retrieved in $O(k)$ time; this structure contains, in fact, the greedy embedding of the current subsequence in the input word, represented succinctly. In this context, we can show the following theorem.
\begin{theorem}
Given a word $w \in \Sigma^*$, with $|w|=n$, $|\Sigma|=\sigma$, and an integer $k \leq n$, we can enumerate all subsequences of length $k$ with $O(n )$ time preprocessing and $O(1)$ time delay.
\end{theorem}

\begin{proof}
This time, we implement more efficiently the strategy of Theorem \ref{lemmaRMQ}. As the respective Lemma was, at its turn, reimplementing the strategy of Lemma \ref{lemmaDFA}, this guarantees the correctness of our approach. 

The novel idea we use here is to maintain the greedy embedding of the currently enumerated subsequence in $w$ in a more succinct way, which allows us to identify quicker the next subsequence to be enumerated, and still do the enumeration in the same order as in the algorithm of Theorem \ref{lemmaRMQ}. This succinct representation is relatively straightforward: as contiguous factors of the currently enumerated subsequence $v$ can also be embedded, in the greedy embedding of $v$ in $w$, into contiguous factors of $w$, we will not maintain this embedding of $v$ as a sequence of separate positions of $w$, but as a sequence of maximal contiguous factors of $w$ (that is, we represent compactly, as a range $[i:j]$, the infixes of the greedy embedding which correspond to the factor $w[i:j]$). We will show, in the following, that this representation allows us to identify in $O(1)$ the longest prefix (i.e., the rightmost position $i_j$) of the greedy embedding $i_1,\ldots, i_k$ of $v$ in $w$ such that the sequence of positions $i_1,\ldots, i_j$ can be extended to obtain a subsequence not enumerated before. A key idea here is that this extension can also be done in $O(1)$: we will perform this extension with a sequence $i'_{j+1},i'_{j+1}+1,\ldots, i'_{j+1}+k-1$ (represented as range $[i'_{j+1}: i'_{j+1}+k-1]$, corresponding to the factor $w[i'_{j+1}: i'_{j+1}+k-1]$ of $w$). We will choose $i'_{j+1}$ exactly in the same way as in Theorem \ref{lemmaRMQ}, which will guarantee that we did not enumerate any subsequence whose greedy embedding starts with $i_1,\ldots, i_j, i'_{j+1}$ until now. Moreover, our algorithm will then enumerate all subsequences whose greed embedding starts with $i_1,\ldots, i_j, i'_{j+1}$, and then moves on and find a position $i''_{j+1}$ and enumerates all subsequences whose greedy embedding starts with $i_1,\ldots, i_j, i''_{j+1}$. This is done, until no such position ($i'_{j+1}$, $i''_{j+1}$, and so on) is found. The implementation of this strategy is, however, a bit more technical than the algorithms in the previous lemmas.

Before starting to describe the details of this strategy, we build some more data structures for the input word $w$. Consider the decomposition of $w$ in runs: $w=a_1^{i_1}a_2^{i_2}\cdots a_k^{i_k}$, where $a_i\neq a_{i+1}$ for $i\in [k-1]$. We define and array $R$, with $k\leq n$ elements, which stores the run length encoding of $w$. Basically, for $i\in [k]$, we want to have $R[i]=(a_i,j_i+1,j_{i+1})$, where $j_i=\sum_{\ell=1,i-1} i_\ell$. Moreover, we also define and array $M$ with $n$ elements, such that for each $i\in [n]$, we have $M[i]=k$ if and only if $R[k]=(a,x,y)$, with $x\leq i\leq y$. 

It is not hard to see that the elements of $R[\cdot]$ and $M[\cdot]$ can be computed in linear time $O(n)$. For simplicity, let $w'=ww[n]$; this simply ensures that we consider a word $w'$ that ends with a repetition, which allows us to avoid corner cases when computing the elements of $M$ and $R$, by making sure that we go through all positions of $w$ (while not considering the last letter of $w'$, which has only a technical role in our computation). We first set $j_1=0$. Then, starting with $i=1$, we increment $i$ and set $M[i]=1$ while $w'[i]=w'[i+1]$ and $i\leq n$. When this loop ends, we set $R[1]=(w'[i], 1, i)$ and $j_2=i$. Then, for some $\ell\geq 2$, we proceed similarly. We start with $i=j_{\ell}+1$, and we increment $i$ and set $M[i]=\ell$ while $w'[i]=w'[i+1]$ and $i\leq n$; when this loop ends, we set $R[\ell]=(w'[i], j_{\ell}+1, i)$ and $j_{\ell+1}=i$.  

Further, we define two further arrays, with each $n$ elements. The first one, $\jumpAhead[\cdot]$, is defined as $\jumpAhead[i]=\min(\{n+1\}\cup \{j\mid j>i+1, w[j] \neq w[i+1]\}$. The second one, $\jumpBack[\cdot]$, is defined as $\jumpBack[i]=\max(\{0\}\cup \{j\mid j<i-1, w[j] \neq w[i-1]\}$. Before showing how they are computing, let us understand the role of such arrays. Assume that we have a greedy embedding $i_1,\ldots,i_\ell$ of some subsequence of length $\ell$ of $w$; then, we can always extend this greedy embedding, and obtain two distinct, new greedy embeddings, by either $i_\ell+1$ or $\jumpAhead[i_\ell]$ (provided that both these positions are at most $n$). Similarly, if the positions of the above greedy embedding of $v$ in $w$ form, in fact, a range $[i_1:i_\ell]$, then $\jumpBack[i_\ell]$ (if greater or equal than $i_1$) gives us a position $i_j$, with $1\leq j\leq \ell$, such that a new greedy embedding, distinct from $i_1,\ldots,i_\ell$, but sharing with it the prefix $i_1,\ldots,i_j$, can be obtained. These arrays can be clearly computed in linear time. We simply explain how this is done for $\jumpAhead[\cdot]$: for some $i\leq n$, let $(a,x,y)=M[i]$. If $y>i$, then $\jumpAhead[i]=y+1$. If $y=i$ and $y<n$, then let $(b,y+1,y')=M[i+1]$, and set $\jumpAhead[i]=y'+1$. Else, set $\jumpAhead[i]=n+1$.

As a final step in this preprocessing, following Observation \ref{preprocessing}, we construct for the word $w$ the array $\prevArray[\cdot]$. Further, as in Theorem \ref{lemmaRMQ}, we construct data structures $RMQ_{\prevArray}$ allowing us to answer range minimum queries for the array $\prevArray[\cdot]$; this can be done in linear time \cite{BenderF02}. 

We can now explain how our enumeration algorithm works.  

Firstly, we define a subroutine $\splitt(a,b,\ell)$, for parameters $1\leq a\leq b\leq n$ and $\ell\leq n$, which will run in constant time. The idea of this subroutine is that we assume that there exists a subsequence $v$ of length $\ell$, whose greedy embedding in $w$ has the form $\alpha, a,a+1,\ldots,b$ (where $\alpha$ is a strictly increasing sequence of positions of $w$). In this subroutine, we compute the largest $j\in [a:b-1]$ (i.e., the largest prefix of the greedy embedding, which ends between $a$ and $b$) such there exists a subsequence $v'$ of $w$ whose greedy embedding is $\alpha, a, a+1, \ldots, j, j', j'+1, \ldots, j'+r$, with $j'> j+1$. Note that such a greedy embedding  $\alpha, a, a+1, \ldots, j, j', j'+1, \ldots, j'+r$ corresponds to a difference subsequence than $\alpha, a,a+1,\ldots,b$, due to the definition of greedy embeddings. 

The algorithm implemented by this subroutine $\splitt(a,b,\ell)$ is the following. If $\jumpAhead[b-1]\leq n- k+\ell$, then $\splitt(a,b,\ell)$ returns $b-1$ (as we can take $j'=\jumpAhead[b-1]$, given that $w[j']\neq w[b]$).  If $\jumpAhead[b-1]> n- k+\ell$, then let $j=\jumpBack[b]$. As  $\jumpAhead[b-1]\leq n- k+\ell$, for every position $g$ with $b-1\geq g>j$ we have that $\jumpAhead[b-1]=\jumpAhead[g]$, so there is no subsequence $u$ of length $k$ of $w$, whose greedy embedding is $\alpha, a, a+1, \ldots, g, g', g'+1, \ldots, g'+f$, with $g'\neq g+1$. Now, we can return $j=\jumpBack[b]$ as the result of $\splitt(a,b,\ell)$ if and only if $n\geq b+k-\ell+1$. Indeed, in that case, we can construct the subsequence $v$ of length $k$ whose greedy embedding is $\alpha, a, a+1, \ldots, j, j+2, j+3, \ldots,  b, b+1, \ldots, b+k-\ell+1 $, and note that $\jumpAhead[j]=j+2$. If we did not return yet any $j$, then we return $-1$.

Once the subroutine is defined, we can proceed with the enumeration algorithm. We maintain two stacks $S$ and $M$; for simplicity, we assume that both are implemented statically (they have at most $n+1$ elements) and there are pointers $top(S)$ and $top(M)$ to the top elements of these two stacks). The stack $S$ contains tuples $([i:j], \beta, I, \ell)$ where $\beta\in \{0,1\}$, $0\leq i\leq j\leq n$, $\ell\leq n$, and $I$ is a queue of disjoint ranges $[i':j']$, with $1\leq i'\leq j'\leq n$ (this queue $I$ has the same role as the queue $I$ defined and used in the proof of Theorem \ref{lemmaRMQ}, for the pair $(j,I)$). Moreover, if the content of $S$ is $(([0:0], \beta_0, I_0, 0), ([i_1:j_1], \beta_1, I_1, \ell_1), \ldots, ([i_t:j_t], \beta_t, I_t, \ell_t))$, where the top of the stack is to the right, then $w[i_1:j_1]\cdots w[i_t:j_t]$ is the currently enumerated subsequence, and $\ell_r=\ell_{r-1}+(j_r-i_r+1)$ for $r\in [t]$ and $\ell_0=0$. As said, the queues $I_r$ are used as in the previous lemma (considered for the pairs $(j_r,I_r)$), while $\beta_r=1$ has the meaning that new subsequences whose greedy embedding has a  prefix $w[i_1:j_1]\cdots w[i_r:g]$ for some $g\leq j_r$ still have to be enumerated. On the other hand, $M$ is stack that contains those positions of $S$ where we have records $([i:j], 1, I, \ell)$ (increasingly ordered bottom to top); in other words, $top(M)$ returns the upmost element of $S$ with the second component equal to $1$. 

In our enumeration, we first define the tuples $([0:0], \beta_0, I_0, 0)$, $([1:k-1], \beta_1, I_1, k-1)$, $([k:k], 0, \emptyset, k)$. If $RMQ_{\prevArray}[2:n-k+1]= 0$, we set $I_0=([2:n-k+1])$ and $\beta_0=1$; this allows us to see if there are subsequences of length $k$ whose greedy embedding does not start with $w[1]$ (and to store this information, so that we can explore these subsequences later). If $RMQ_{\prevArray}[k+1:n]\leq k-1$, we set $I_1=([k+1:n])$ and $\beta_1=1$; this allows us to see if there are subsequences of length $k$ whose greedy embedding has the form $w[1:k-1]w[f]$ for some $f\neq k$ (and store this information, so that we can explore them later). If we did not already set $\beta_1=1$, then we compute $j=\splitt(1,k-1,k-1)$; if $j\neq -1$, we set $\beta_1=1$. Once the tuples are computed, we set $S=([0:0], \beta_0, I_0, 0), ([1:k-1], \beta_1, I_1, k-1), ([k:k], 0, \emptyset, k))$, and define $M$ accordingly (to contain the positions of $S$ where elements whose second component is $1$ are stored). We output $(1,k)$, representing the subsequence $w[1:k]$, the first in our enumeration. 

In a general step in the enumeration, we proceed as follows. Set $top(S)=top(M)$. If the stacks are empty, we stop. Note, at this point, that setting $top(S)=top(M)$ corresponds to removing some elements of $S$. However, this is not done explicitly: we simply move the top-pointer for $S$ to the desired position in $O(1)$ (possible due to the static implementation of the stacks), and then we overwrite the elements that were above this position before.

Futher, we pop the element $([i:j],1,I,\ell)$ from the top of $S$ (as well as the top element of $M$). We now have to do a case analysis. 

If $I\neq \emptyset$, the analysis is very similar to that implemented in the algorithm of Theorem \ref{lemmaRMQ}. We extract the first interval $[c:d]$ from the queue $I$. Let $g=RMQ_{\prevArray}[c:d]$. Let $g_1=RMQ_{\prevArray}[c:g-1]$ and $g_2=RMQ_{\prevArray}[g+1:d]$. If $g_1\leq j$, insert $[c:g-1]$ in $I$; if $g_2\leq b$, insert $[g+1:d]$ in $I$. Let $I'$ be the queue obtained from $I$ after these potential insertions. We now consider the tuples $([i:j],\beta',I',\ell)$, $([g:g'-1],\beta'', I'', k-1)$, $([g':g'],0, \emptyset', k)$; here $g'=g+(k-\ell-2)$. If $I'\neq \emptyset$, we set $\beta'=1$; if $I'=\emptyset$ and $\splitt(i,j,\ell)\neq -1$, we set $\beta'=1$. Further, if $RMQ_{\prevArray}[g':n]\leq g'-1$, then we add $[g':n]$ to $I''$ and set $\beta''=1$; if $RMQ_{\prevArray}[g':n]> g'-1$, we set $\beta''=1$ when $\splitt(g,g'-1,k-1)\neq -1$, and set $\beta''=0$ otherwise. We insert $([i:j],\beta',I',\ell)$, $([g:g'-1],\beta'', I'', k-1)$, $([g':g'],0, \emptyset', k)$, in this order, in $S$, and update the top elements of $M$ accordingly. We output: delete the suffix of length $k-\ell$ of the previous subsequence, add suffix $(g,g')$. This case was very similar to the first step of our enumeration. 

If $I= \emptyset$, we compute $h=\splitt(i,j,\ell)$, and let $g=\jumpAhead[h]$. We now consider the tuples $([i:h],\beta',I',\ell')$, $([g:g'-1],\beta'', I'', k-1)$, $([g':g'],0, \emptyset', k)$; here $\ell'=\ell-(b-h)$ and $g'=g+(k-\ell'-2)$. 
If $RMQ_{\prevArray}[g+2:n-\ell'+1] \leq h$, we set $I'=[g+2:n-\ell'+1]$ and $\beta'=1$; otherwise, we compute $e=\splitt(i,h,\ell')$ and set $\beta'=1$ if $e\neq -1$ and $\beta'=0$, otherwise.
Further, if $RMQ_{\prevArray}[g':n]\leq g'-1$, then we add $[g':n]$ to $I''$ and set $\beta''=1$; if $RMQ_{\prevArray}[g':n]> g'-1$, we set $\beta''=1$ when $\splitt(g,g'-1,k-1)\neq -1$, and set $\beta''=0$ otherwise. 
We insert $([i:h],\beta',I',\ell')$, $([g:g'-1],\beta'', I'', k-1)$, $([g':g'],0, \emptyset', k)$, in this order, in $S$, and update the top elements of $M$ accordingly. We output: delete the suffix of length $k-\ell'$ of the previous subsequence, add suffix $(g,g')$. Note that, at this point, if desired by the used, we can output (based on the content of $S$) the currently enumerated subsequence, as explained when the stack $S$ was introduced.

This concludes the description of our algorithm. 

From the above, it is clear that the delay between two outputs is constant. Now, for the correctness, we note that the first output corresponds to the first output of the algorithm in Theorem \ref{lemmaRMQ}. Assume that the first $i$ outputs of the current algorithm correspond, respectively, to the first $i$ outputs of the algorithm of Theorem \ref{lemmaRMQ}. In both cases, we then traverse right to left the subsequence we output last, and find the longest prefix of the respective subsequence which can be extended in a different way than in the previously enumerated subsequences, to obtain a novel subsequence of length $k$ which was not enumerated yet. The only major difference is that in this algorithm we identify this prefix in $O(1)$ time: the stack $M$ allows us to find in $O(1)$ the last contiguous factor of the greedy embedding of the currently enumerated subsequence which either can be continued with a different letter than in the current subsequence, or contains a position from which we can continue differently. This element of $S$ is pointed by $top(M)$. Note that we precompute this kind of information for all elements of $S$, every time we introduce an element in $S$ -- they are simply those elements whose $\beta$-component is set to $1$. As said, we can also directly jump to upmost element of $S$ which has the $\beta$ component equal to $1$, using the stack $M$. Finally, we can also do the extension in $O(1)$ time: we just have to add two contiguous factors to $S$, at most, and update $M$. As a side note, when doing an extension with the factor $w[a:b]$, we split it as $w[a:b-1]$ as two factors $w[b]$ for technical reasons, as it allows us to have an uniform procedure in identifying the longest prefix of the current greedy embedding which can be extended differently. The correctness, thus, follows: our current algorithm is simply a faster implementation of the previously presented algorithms, from Theorem \ref{lemmaRMQ}.\qed
 \end{proof}

	\section{Conclusion}
	In this work, we gained a better understanding of the set of length $k$ subsequences of a given word. Based on the notions of $m$-nearly $k$-universality and absent subsequences, 
we provided the minimal and maximal number of 
absent subsequences for a given alphabet $\Sigma$, subsequence length $k$ and universality index $\iota$. During the investigation we gave two tight bounds for $m$ as well as 
constructions on $\iota$-universal words over $\Sigma$ that omit 
this number of absent subsequences of length $k$. With these 
results we obtained a further restriction on the index of Simon's Congruence. 
While absent subsequences were shown to be instrumental in gaining a better understanding in the set of subsequences of given length of a word from a combinatorial perspective (both here and in works such as \cite{DBLP:journals/fuin/KoscheKMS22,DBLP:journals/tcs/FleischmannHHHMN23}), it also makes sense to ask whether we can effectively produce the respective set of present subsequences, by an efficient algorithm. We prove that this set can be enumerated optimally, i.e., we give an algorithm with linear preprocessing and linear output delay, or with constant delay in the case, when we only output some initial subsequence and short sketches showing how the current subsequence can be produced from the previous one.
While minimal and maximal bounds for the number of absent subsequences $m$ are shown, for given universality $\iota$ and subsequence length $k$, it remains open to determine all possible values of $m$, the structure and count of corresponding words, and whether efficient algorithms for obtaining them exist.

	

	\bibliographystyle{splncs04}
	\bibliography{refs}
	
	
	
	
	
	
\end{document}